\documentclass[12pt]{iopart}

\usepackage{lineno,hyperref}
\usepackage{bm}
\usepackage{amsmath}
\usepackage{amssymb}
\usepackage{amsthm}
\usepackage{iopams}
\usepackage{graphicx}
\usepackage[dvipsnames]{xcolor}
\usepackage{siunitx}
\usepackage{algorithm}
\usepackage{algpseudocode}
\usepackage{calrsfs}
\DeclareMathAlphabet{\pazocal}{OMS}{zplm}{m}{n}

\newtheorem{theorem}{Theorem}[section]
\newtheorem{lemma}[theorem]{Lemma}

\newtheorem{proposition}[theorem]{Proposition}
\newtheorem{definition}[theorem]{Definition}
\DeclareMathOperator{\spark}{spark}

\modulolinenumbers[5]

\definecolor{editColor}{RGB}{0, 102, 204}

\newcommand{\argmax}[2]{\underset{#2}{\mathrm{arg\: max}}\left\lbrace #1\right\rbrace}

\begin{document}


\title[Unbiased Parameter Estimation and Signal Reconstruction]{On Unbiased Parameter Estimation and Signal Reconstruction}

\author[inst]{Joonas Lahtinen}
\address{Mathematics \& Computing Sciences, Tampere University, Tampere 33720, Finland}
\ead{joonas.j.lahtinen@tuni.fi}

\begin{abstract}
In this paper, we extend the theory of depth-unbiased source localization to unbiased parameter estimation and signal reconstruction for an arbitrary number of non-zero parameters. The topic touches on exact reconstructibility, most commonly studied in compressed sensing and multisource estimation across various imaging problems. The theoretical results derive upper bounds on the number of recoverable parameters in the noiseless case, and define a probability measure to assess the likelihood of recovering all non-zero parameters with correct magnitude order. The work provides a mathematical explanation of the open question regarding the noise robustness of standardized and unbiased methods. The paper also reveals a trade-off between the number of sensors and the signal-to-noise ratio. Numerical experiments demonstrate the theoretical findings.
\end{abstract}

\noindent{\it Keywords\/}: Linear forward problem, Parameter estimation, Signal reconstruction, Solution bias, Compressed sensing

\submitto{\IP}
\maketitle



\section{Introduction}
In inverse problems where the objective is to recover data that explain parameters or to reconstruct a signal or image from indirect observations, one can ask whether it is possible to exactly reconstruct the source from limited data. On some occasions, when the observed signal depends on the source type and the distance between the source and the sensor, an estimation bias may arise. Here, following the notion of bias determined in \cite{MNE_cardio_biased,PascualMarqui2002} we define it as follows:
\begin{definition}\label{def:orgUnbias}
    A method is {\bf unbiased} if no solution of $A{\bf x}={\bf y}$ is favored in the decision process.
\end{definition}

In the usual {\em decision process} used to interpret inversion results, the highest estimation magnitudes are thought to contain the true source \cite{Manfred1999,SekiharaNagarajan2008,SamuelssonJohnMaximumPrinciple4DistrMethds2020}.

A biased estimate arises from cheap but unrealistic solutions of a given optimization problem that reduce the data error, i.e., the likelihood in the Bayesian sense. The bias could be reduced by introducing a highly sophisticated regularization scheme or a priori; however, when the number of model parameters to be estimated is high, as in the reconstruction of high-resolution images or the neuroelectromagnetic inverse source problem \cite{Grech2008,Calvetti2009,GravedePeralta009NEIP}, electrocardiography inverse problem \cite{GuofaShou2007StEI}, hydrogeology \cite{CarreraJess2005hydrogeology}, and electrical resistance tomography \cite{JohnsonGeoPhys2015,ShiWenyangGeoPhys2023}, advanced prior models are not computationally feasible, even when Monte Carlo methods are employed. 

A technique called {\em standardization} \cite{PascualMarqui2002} was first introduced for the static neuroelectromagnetic inverse source problem, which aims to determine the location, magnitude, and orientation of the source, and was later generalized to localization problems arising from the Poisson equations with Neumann boundary conditions \cite{Lahtinen2024}, as well as to time-varying localization problems \cite{Lahtinen2024SKF}. The technique uses Gaussian random fields to infer the source's unbiased location, i.e., unbiased depth. In practice, it is a linear transformation from biased source estimation to depth-unbiased Z-scores \cite{WagnerMichael2004sLORscr,LealAlberto2008sLORscr}. As an unbiasing strategy, it provides unbiased localization only if the observations are noiseless, and it can find only one source, as the location is determined to be the one with the highest Z-score. However, in real-life applications, we do not focus on the maximum Z-score but on a percentage of the top values. For example, the rule of the highest 30 \% of Z-score values has been proposed \cite{LiChunsheng2016_30threshold}. This approach accounts for uncertainties from modeling errors, approximation errors, and observation noise. 

Standardized methods have been among the most successful for localizing the epileptic focus in the human brain \cite{Dumpelmann2012sLORepil,deGooijer2013_14,Coito2019_18,LiRui2021_7,KimKwangYeon2022sLORepil,vandeVelden2023}. Partially, the success can be explained by the unbiased localization that allows the source to be estimated everywhere within the given domain, but standardized methods are also highly robust against additive measurement noise due to yet an unknown reason. An obvious limitation of the technique is that it recovers one of the parameters needed to define the source.

In this paper, we further generalize the transformation to yield unbiased sparse estimates by considering unbiased parameters rather than a single parameter, the location. We take a particular interest in models in which multiple parameters define a source, as well as in multisource scenarios. We address the conditions required for exact source and multisource reconstructions in the noiseless case and in the presence of observation noise. The findings of this study provide a complete explanation of why the standardization technique yields unbiased, noise-robust estimates and why it sometimes seems beneficial to use fewer sensors to estimate the source. In the experiments, we demonstrate the properties of unbiased estimation in compressed sensing and 2- and 3-dimensional electrical source imaging, using a realistic multicompartment head model in the 3-dimensional case. 

\section{Preliminary}
We consider the following observation model:
\begin{equation}\label{eq:BasicFrwrd}
    {\bf y}=L{\bf x},
\end{equation}
where observations ${\bf y}\in\mathbb{R}^m$, the lead field or system matrix $L\in \mathbb{R}^{m\times dn}$ is built as
\begin{equation}
    L=\begin{bmatrix}
        L_1&\cdots&L_n
    \end{bmatrix},
\end{equation}
where any submatrix $L_k$, $k=1,\cdots, n$ belongs to $\mathbb{R}^{m\times d}$, where $d$ is the number of parameters describing a source and $n$ is the number of possible distinct sources that are subvectors ${\bf x}_k\in\mathbb{R}^d$ listed in ${\bf x}\in\mathbb{R}^{dn}$. There are a few necessary assumptions we make here so that it makes sense to talk about exact solutions in any meaningful way:
\begin{enumerate}
    \item Every observation ${\bf y}$ belong to the range of the system matrix $L$ excluding zero vector, i.e., ${\bf y}\in\pazocal{R}(L)$. This will also include the random noise added to the observations.
    \item No block $L_k$ is zero matrix.
\end{enumerate}
Since the estimation $\hat{\bf x}$ is obtained here by minimizing the squared error between $L\hat{\bf x}$ and ${\bf y}$, the part of the observations that does not belong to the range of the system matrix cannot be altered by any change in $\hat{\bf x}$, nor does the minimum argument change in terms of observation vectors outside the range of the system matrix, which explains the rationale of item (1). On the application side, we assume no modeling or approximation errors in the system matrix. In the case of item (2), the part of the estimation vector that would correspond to a zero block or zero column in the system matrix would be effectively zero; therefore, the edge cases introduced by zero columns or blocks are redundant, which is why we neglect them right away. 

In the paper \cite{Lahtinen2024}, it has been shown that the definition of unbiasedness (Definition \ref{def:orgUnbias}) under the decision based on the maximum magnitude of the Z-score variables ${\bf z}$, especially focusing on the {\em favoring} of a solution, can be mathematically expressed as follows:

\begin{definition}\label{def:unbias}
    Let $\hat{\bf x}_1,\:\hat{\bf x}_2\in\mathbb{R}^{m\times d}$ be two distinct and equally acceptable estimates for the linear equation ${\bf y}=A{\bf x}$ such that the norm of the difference $\left\|{\bf y}-A{\bf x}\right\|_2$ --- or likelihood in case of Gaussian models --- is equal. Moreover, let $T$ be the transformation from estimates to decision variables. If there exists $a\neq 1$ such that $\left\|T\left[\hat{\bf x}_1\right]\right\|_2=a\left\|T\left[\hat{\bf x}_2\right]\right\|_2$, the method providing these estimates is {\bf biased}. In contrast, if $a=1$, the method is {\bf unbiased}.
\end{definition}

In the same paper, it has been proved that in the one-dimensional case, where $d=1$, one source can be recovered exactly from noiseless observations by 
\begin{equation}\label{eq:sLOR1}
    k=\argmax{\frac{{\bf L}_i^T\Sigma^{-1}{\bf y}}{\sqrt{{\bf L}_i^T\Sigma^{-1}{\bf L}_i}}}{i=1,\cdots,n},
\end{equation}
where $\Sigma=LL^T+C$, and $C$ is the observation noise covariance stemming from the assumed Gaussian likelihood model,
and computing
\begin{equation}\label{eq:sLOR2}
    \hat{ x}_k={\bf L}^\dagger_k{\bf y},
\end{equation}
where $(\cdot)^\dagger$ is the pseudoinverse. Now, by simply expanding the index to a parameter variable $\bm{\theta}_k\in \mathbb{S}^{d-1}\subset \mathbb{R}^d$ such that
\begin{equation}
    (\hat{\bm{\theta}},k)=\argmax{\frac{\bm{\theta}^TL_i^T\Sigma^{-1}{\bf y}}{\sqrt{\bm{\theta}^T L_i^T\Sigma^{-1}L_i\bm{\theta}}}}{i=1,\cdots,n,\: \bm{\theta}\in\mathbb{S}^{d-1}},
\end{equation}
the estimate $\hat{\bf x}_k=\left\|L_k^\dagger{\bf y}\right\|_2\hat{\bm{\theta}}$ end up being biased. This can be shown to be true by denoting $\hat{\bf y}=\Sigma^{-1/2}{\bf y}$ and using the compact singular value decomposition $\Sigma^{-1/2}L_k=\check{U}_k\check{S}_k\check{V}_k^T$, where $\check{S}_k$ is a square diagonal matrix with non-zero diagonals. Since there exists ${\bf v}\in \mathbb{S}^{d-1}$, i.e., on the $d$-dimensional hypersphere, such that $\bm{\theta}=\check{V}_k{\bf v}$, hence, the solution can be written as $\check{U}_k\check{S}_k{\bf v}=r\hat{\bf y}$, for some $r\in \mathbb{R}$. Because the solution depends on the nonzero entries of $\check{S}_k$, it is biased. Namely, the solution is biased on the normalized parameter space $\mathbb{S}^{d-1}$. However, this consideration reveals to us that unbiased solutions can be obtained by the following weighting of the individual source estimates:
\begin{equation}
\begin{split}
    {\bf z}_k&=(\check{V}_k\check{S}_k\check{V}_k^T)^{-1}\check{V}_k\check{S}_k\check{U}_k^T\hat{\bf y}=\check{V}_k\check{S}_k^{-1}\check{V}_k^T\check{V}_k\check{S}_k\check{U}_k^T\hat{\bf y}\\
    &=\check{V}_k\check{U}_k^T\hat{\bf y}.
\end{split}
\end{equation}
By denoting $A_k=\check{U}_k\check{V}_k^T$, we obtain
\begin{equation}
    {\bf z}_k=A_k^\dagger\hat{\bf y}.
\end{equation}
Moreover, Theorem 4.1 in \cite{Lahtinen2024}, shows that if there is an index $k$ such that the observations are ${\bf y}=L_k{\bf x}_k$ for any ${\bf x}_k\in \mathbb{R}^d$, the highest Z-score is obtained for all vectors from the preimage of $L_k$. Now, by writing the minimization problem used in the proof with $A$:
\begin{equation}
    \textnormal{maximize}\: {\bf z}^TA^T\hat{\bf y}\quad \textnormal{subject to }\sqrt{{\bf z}^TA^TA{\bf z}}=c,
\end{equation}
for some $c>0$, we get the same solution for ${\bf z}_k$ once the observations comes from the range of $\Sigma^{-1/2}L_k$. The next theorem shows that the estimate is indeed unbiased.
\begin{proposition}
    Let 
    \begin{equation}
        A=\begin{bmatrix}
            \check{U}_1\check{V}_1^T & \cdots & \check{U}_n\check{V}_n^T
        \end{bmatrix},
    \end{equation}
    where the submatrices originate from the compact singular value decompositions of modified sub-lead fields $\Sigma^{-1/2}\left[LP\right]_k=\check{U}_k\check{S}_k\check{V}_k^T$, where $\left[LP\right]_k\in\mathbb{R}^{m\times d}$, $\Sigma=LPL^T+C$ with source covariance $P\in\mathbb{R}^{nd\times nd}$, and measurement covariance $C\in\mathbb{R}^{m\times m}$. Then
    \begin{equation}
        {\bf z}_k=A_k^\dagger\Sigma^{-1/2}{\bf y}
    \end{equation}
    gives unbiased estimation for the problem ${\bf y}=L{\bf x}$ if there is one source generating the observations ${\bf y}$.
\end{proposition}
\begin{proof}
    Assume that the source has magnitude $a>0$ and parameters $\bm{\theta}\in \mathbb{R}^d$ and is at the $k$th position, then the modified observations are $\Sigma^{-1/2}{\bf y}=a\Sigma^{-1/2}\left[LP\right]_k\bm{\theta}=a\check{U}_k\check{S}_k\check{V}_k^T\bm{\theta}$. Then for arbitrary sub-matrix denoted as $\Sigma^{-1/2}\left[LP\right]_j=\check{U}_j\check{S}_j\check{V}_j^T$, we observe an estimation sub-vector
    \begin{eqnarray*}
        {\bf z}_j=A_j^\dagger \Sigma^{-1/2}{\bf y}=a\check{V}_j\check{U}_j^T\check{U}_k\check{S}_k\check{V}_k^T\bm{\theta}.
    \end{eqnarray*}
    Focusing on the matrix product $\check{U}_j^T\check{U}_k$, we notice that
\begin{equation}
    \left\|\check{U}_j^T\check{U}_k\bm{\theta}\right\|_2^2=\bm{\theta}^T\check{U}_k^T\check{U}_j\check{U}_j^T\check{U}_k\bm{\theta}
\end{equation}
for any vector $\bm{\theta}\in\mathbb{S}^{d-1}$. From here, we notice that $\check{U}_j\check{U}_j^T$ is an orthogonal projection matrix, as it is (1) symmetric: $\left(\check{U}_j\check{U}_j^T\right)^T=\check{U}_j\check{U}_j^T$ and (2)  idempotent: 
\begin{equation}
    \left(\check{U}_j\check{U}_j^T\right)^2=\check{U}_j\check{U}_j^T\check{U}_j\check{U}_j^T=\check{U}_jI_m\check{U}_j^T=\check{U}_j\check{U}_j^T.
\end{equation}
Now, by applying Pythagoras' theorem to $\check{U}_k\bm{\theta}$:
\begin{equation}
    \left\|\check{U}_k\bm{\theta}\right\|_2^2=\left\|\check{U}_j\check{U}_j^T\check{U}_k\bm{\theta}\right\|_2^2+ \left\|(I-\check{U}_j\check{U}_j^T)\check{U}_k\bm{\theta}\right\|_2^2
\end{equation}
from which we can conclude
\begin{equation}
    \left\|\check{U}_k\bm{\theta}\right\|_2\geq \left\|\check{U}_j\check{U}_j^T\check{U}_k\bm{\theta}\right\|_2=\bm{\theta}^T\check{U}_k^T\check{U}_jI_m\check{U}_j^T\check{U}_k\bm{\theta}.
\end{equation}
This further concludes that $\left\|\bm{\theta}\right\|_2=\left\|\check{U}_k\bm{\theta}\right\|_2\geq \left\|\check{U}_j^T\check{U}_k\bm{\theta}\right\|_2$. Yielding $\left\|{\bf z}_k\right\|_2\geq \left\|{\bf z}_j\right\|_2$.
    
    Moreover,  let there be a sub-matrix such that $\Sigma^{-1/2}\left[LP\right]_j=\check{U}_k\check{S}_j\check{V}_j^T$, then, on one hand 
    \begin{eqnarray*}
        {\bf z}_j=A_j^\dagger \Sigma^{-1/2}{\bf y}=a\check{V}_j\check{S}_k\check{V}_k^T\bm{\theta},
    \end{eqnarray*}
    and on the other hand ${\bf z}_k=a\check{V}_k\check{S}_k\check{V}_k^T\bm{\theta}$. This yields $\left\|{\bf z}_j\right\|_2=\left\|{\bf z}_k\right\|_2$, thus proving that the solution is unbiased according to Definition \ref{def:unbias}.
\end{proof}

The same modified pseudoinverse can be achieved by a slight modification of the construction of 3-dimensional standardization weighting given by Pascual-Marqui in \cite{Pascual2007sqrtm}:
\begin{equation}
    (L_k^T\Sigma^{-1}L_k)^{-1/2}L_k^T\Sigma^{-1/2}=\left(\check{V}_k\check{S}_k^{2}\check{V}_k\right)^{-1/2}\check{V}_k\check{S}_k\check{U}_k^T=\check{V}_k\check{U}_k^T=A_k^\dagger,
\end{equation}
where now $\Sigma=LPL^T+C$ instead of formulation $\Sigma^\dagger = \left(LL^T+\alpha (I-m^{-1}{\bf 1}{\bf 1}^T)\right)^\dagger$ for neuroelectric inverse problem and $\Sigma = \left(LL^T+\alpha I\right)$ for neuromagnetic inverse problem, where $\alpha>0$ is a regularization term. When we assume full column rank for the system $A\in\mathbb{R}^{m\times K}$, where $K\leq m$, 
the expected standardized metric (expectation of Z-scores) attains a global maximum in the $A$, which equivalently means that the expected residual variance $\left\|\left(\Sigma^{1/2}\right)^\dagger({\bf y}-A{\bf x})\right\|_2^2$ attains a global minimum, as proved in \cite{HoltershinkenOnEquivalence}.

Therefore, the new idea provided in this paper is to construct the unbiased system matrix consisting of the blocks $\check{U}_k\check{V}_k^T$ via singular value decomposition instead of standardization weighting involving matrix square roots computed using the Schur method \cite{BjorckA1983ASqrtm}. Because of our formulation, we can use results from compressed sensing to explain the unbiased estimation and its limits.

From the proposition above, we note that instead of weighting source estimates, we can use the modified forward matrix $A$ to derive the unbiased estimates. Moreover, we conclude that the matrices $\check{U}_k$ are most important for achieving unbiased estimation, as the matrices $\check{V}_k$ introduce isometric mappings between some normalized parametrization on $\mathbb{S}^{d-1}$. In addition, we have a linear and invertible map $T_j=\check{S}_j\check{V}_j^T$ to map each source $a_j\bm{\theta}_j$ to unbiased estimates ${\bf z}_j$. Therefore, we can now focus on solving the following equation 
\begin{equation}
    \hat{\bf y}=\mathcal{U}{\bf z},
\end{equation}
where
\begin{equation}
    \mathcal{U}=\begin{bmatrix}
        \check{U}_1&\cdots & \check{U}_n
    \end{bmatrix}\in \mathbb{R}^{m\times dn}.
\end{equation}

\section{Multisource estimation}\label{sc:multisourceest}

In this section, we connect unbiased estimation via $\hat{\bf y}=\mathcal{U}{\bf z}$ to compressed sensing, namely the search for the sparsest possible solution to a linear system. In fact, we are seeking a sparse solution $\hat{\bf z}$ constructed from multiple distinct sources, i.e., $\left\|\hat{\bf z}\right\|_0=Nd$. This kind of solution is called {\em compressed} because the observation signal can be expressed by a sparse vector. This type of inverse problem arises when there are multiple sources, the sources have a spatial distribution modeled as active neighboring sources, or when reconstructing images from incomplete data. In this section, we aim to determine how many active elements (e.g., sources or pixels) we can reconstruct exactly. Notably, the investigation of unique sparsest solutions reduces to analyzing the orthogonality of the normalized columns of the system matrix, which closely parallels the proposed framework for unbiased estimation. Fundamental questions of compressed sensing are the existence and uniqueness of the sparse solution. So, instead of asking which solution out of infinitely many is the correct one, we ask which solution is the sparsest. Here, we explore important results on the existence of a unique sparsest solution, which we extend to noisy cases in the next section. 

We next define one of the key concepts of compressed sensing: mutual coherence.

\begin{definition}\label{def:MutCoh}
    Let $\mathcal{U}\in\mathbb{R}^{m\times n}$ be a matrix whose columns are unit vectors $\lbrace{\bf U}_k\rbrace$. {\bf Mutual coherence} $\mu(\mathcal{U})$, as the measure of the maximum similarity between any two columns of a matrix, is defined as
    \begin{equation}
        \mu(\mathcal{U}) = \max_{1\leq s\neq t\leq n}\left|{\bf U}_s^T{\bf U}_t\right|
    \end{equation}
\end{definition}
The following important lemma, stated by L. R. Welch for mutual coherence, lays the foundations for exact multisource estimation.
\begin{lemma}\label{lem:MutCohLowBound}
    (Coherence lower bound \cite{WelchL1974MutualCoherenceBound}) Mutual coherence has a lower bound
    \begin{equation}
        \mu(\mathcal{U})\geq \sqrt{\frac{n/m-1}{n-1}}.
    \end{equation}
\end{lemma}

For later proofs and definitions regarding sparse solutions, we must introduce the concept of vector support.
\begin{definition}
    {\bf Support} of vector ${\bf x}\in\mathbb{R}^n$ indicates the indices of the vector that are non-zero, and it is denoted as 
    \begin{eqnarray*}
        \mathrm{supp}({\bf x})=\lbrace k\in \lbrace 1,\dots n\rbrace\colon x_k\neq 0\rbrace.
    \end{eqnarray*}
\end{definition}

For the next theorem, we need the concept of the spark of the matrix. The spark characterizes the smallest support size of a nontrivial vector in the null space of a matrix and therefore plays an important role in assessing the uniqueness and identifiability of sparse representations.

\begin{definition}
    The {\bf spark} of matrix $A\in\mathbb{R}^{m\times n}$ defines the smallest number of columns that are linearly dependent. Formally, this is expressed as
    \begin{eqnarray}
        \spark(A)=\min_{{\bf x }\in \mathbb{R}^n\setminus\lbrace{\bf 0}\rbrace}\left\|{\bf x}\right\|_0\quad\textnormal{subject to}\quad A{\bf x}={\bf 0}.
    \end{eqnarray}
\end{definition}

Hence, we can form the following condition for the unique sparsest possible solution.

\begin{theorem}\label{thm:exactreclimit}
    (Uniquely sparsest solution via spark \cite{Donoho2003MutualCoherenceUniqBound}) If a solution ${\bf x}$ for linear system $A{\bf x}={\bf y}$ satisfies
    \begin{equation}
        \left\|{\bf x}\right\|_0<\frac{\spark(A)}{2},
    \end{equation}
    then this solution is unique and the sparsest possible solution.
\end{theorem}
\begin{proof}
    Let the linear system $A{\bf x}={\bf y}$ have distinct solutions $\hat{\bf x}_0$ and $\hat{\bf x}_1$. So, $A(\hat{\bf x}_1-\hat{\bf x}_0)={\bf 0}$ and hence, $\hat{\bf x}_1-\hat{\bf x}_0$ belong to the null space of $A$. By the definition of spark
    \begin{eqnarray*}
        \left\|\hat{\bf x}_1-\hat{\bf x}_0\right\|_0\geq \spark (A).
    \end{eqnarray*}
    Now, let $S_0=\mathrm{supp}(\hat{\bf x}_0)$ and $S_1=\mathrm{supp}(\hat{\bf x}_1)$. Since $\mathrm{supp}(\hat{\bf x}_1-\hat{\bf x}_0)\subseteq S_0\cup S_1$, we obtain
    \begin{eqnarray*}
        \spark(A)\leq \left\|\hat{\bf x}_1-\hat{\bf x}_0\right\|_0\leq \left|S_1\cup S_0\right|\leq \left\|\hat{\bf x}_1\right\|_0+\left\|\hat{\bf x}_0\right\|_0.
    \end{eqnarray*}
    Thus $\left\|\hat{\bf x}_1\right\|_0\geq \spark(A)-\left\|\hat{\bf x}_0\right\|_0$, and the inequality of the statement implies $\spark(A)-\left\|\hat{\bf x}_0\right\|_0>\left\|\hat{\bf x}_0\right\|_0$. Combining these two inequalities yields $\left\|\hat{\bf x}_1\right\|_0>\left\|\hat{\bf x}_0\right\|_0$.

    If we assume that $\left\|\hat{\bf x}_1\right\|_0<\spark(A)/2$, then $\spark(A)-\left\|\hat{\bf x}_1\right\|_0>\left\|\hat{\bf x}_1\right\|_0$. Taking into account the previously derived inequality $\left\|\hat{\bf x}_1\right\|_0\geq \spark(A)-\left\|\hat{\bf x}_0\right\|_0$, we get $\left\|\hat{\bf x}_0\right\|_0>\left\|\hat{\bf x}_1\right\|_0$. This contradicts the outcome $\left\|\hat{\bf x}_1\right\|_0>\left\|\hat{\bf x}_0\right\|_0$, which means that $\hat{\bf x}_0$ must be uniquely sparsest possible solution.
\end{proof}

The theorem above simultaneously limits how many sources can be reconstructed exactly, since it requires solution uniqueness. The bound can, however, be improved by using the mutual coherence (Definition~\ref{def:MutCoh}):

\begin{theorem}\label{thm:SparsestUniqMutCoh}
    (Sparsest unique solution with mutual coherence \cite{ZhaoYun-Bin2013NewSparsestBounds}) If a solution ${\bf x}$ for linear system $A{\bf x}={\bf y}$ admits to the bound
    \begin{equation}
        \left\|{\bf x}\right\|_0<\frac{1}{2}\left(1+\frac{1}{\mu(A)}\right),
    \end{equation}
    for the mutual coherence $\mu(A)$, this solution is necessarily the sparsest possible.
\end{theorem}

\section{Estimation with noisy data}\label{sc:NoisyEstimation}
In this section, we denote random variables via tilde-accenting, i.e., let the random variable $\tilde{\bf x}$ follow a Gaussian distribution $\pazocal{N}({\bf 0},P)$ and measurement noise follows a Gaussian, i.e., $\tilde{\bf y}\sim\pazocal{N}({\bf 0}, C)$. Then, the measurements are distributed as

\begin{equation}
    p(\tilde{\bf y})=\int_\Omega \pazocal{N}(L{\bf x},C)\pazocal{N}({\bf x}\mid {\bf 0},P)\,\mathrm{d}{\bf x}=\pazocal{N}({\bf 0}, L\Gamma L^\mathrm{T}+C)=\pazocal{N}({\bf 0}, \Sigma).
\end{equation}
Hence, the modified observations $\hat{\bf y}=\Sigma^{-1/2}\tilde{\bf y}$ follows standard Gaussian $\pazocal{N}({\bf 0},I)$. Since the columns of the forward model $\mathcal{U}$ are $m$-dimensional unit vectors, if $\mathcal{U}\in\mathbb{R}^{m\times m}$ is invertible, the column that explains the observations is uniformly random. This follows from the fact that if we express the observation space in a hyperspherical coordinate system, the radial distance is $\chi$-distributed while the direction on $\mathbb{S}^{m-1}$ is uniformly distributed.

The following theorem sets the limits for unbiased multisource estimation with noisy data in the majority of the cases:
\begin{theorem}
    Let ${\bf y}=A\hat{\bf x}$, where $\hat{\bf x}$ is the unique sparsest possible solution and $A\in \mathbb{R}^{m\times n}$ for $n\geq m$. This exact solution $\hat{\bf x}$ obtained from the noiseless data cannot be recovered from the noise-contaminated equation $\hat{\bf y}+\tilde{\bf n}=A{\bf x}$, where $\tilde{\bf n}$ follows a Gaussian distribution, is not possible by any optimization problem of the following form:
    \begin{equation*}
        \left\|{\bf y}-A{\bf x}\right\|_2^2+\alpha\left\|{\bf x}\right\|_p^p,
    \end{equation*}
    for any $p\geq 0$ and $\alpha\geq 0$.
\end{theorem}
\begin{proof}
Let us assume the noiseless case has the unique sparsest possible solution $\hat{\bf x}$ with $\left\|\hat{\bf x}\right\|_0=N$ for ${\bf y}=A{\bf x}$ such that we can denote ${\bf y}=A_N{\bf x}_N+B{\bf 0}$, where $A_N\in\mathbb{R}^{m\times N}$ and $B\in\mathbb{R}^{m\times (n-N)}$. First, let us assume $\tilde{\bf n}_\perp \notin \pazocal{R}(A_N)$, to observe the exact solution we must have
\begin{eqnarray*}
    {\bf y}+\tilde{\bf n}_\perp=A_N\hat{\bf x}_N+B\hat{\bf x}_{n-N}={\bf y}+B\hat{\bf x}_{n-N},
\end{eqnarray*}
yielding $\tilde{\bf n}_\perp \notin \pazocal{R}(B)$ for $p=0$, $p=2$, and $\alpha=0$ for any $p\geq 0$. For the rest of the cases, we can explore the stationary condition provided by the object functions' subgradient form of the optimality condition:
\begin{eqnarray*}
    \left\lbrace 2A_N^T\tilde{\bf n}_\perp- \alpha p\begin{bmatrix}
        \mathrm{sign}(\hat{x}_{N;1})\left|\hat{x}_{N;1}\right|^{p-1}\\
        \vdots \\
        \mathrm{sign}(\hat{x}_{N;N})\left|\hat{x}_{N;N}\right|^{p-1}
    \end{bmatrix}\right\rbrace \ni {\bf 0},
\end{eqnarray*}
that has a singleton representation, as each $\hat{x}_{N;k}\neq 0$. This yields
\begin{eqnarray*}
    0=2{\bf y}^T\tilde{\bf n}_\perp=2\hat{\bf x}_N^TA_N^T\tilde{\bf n}_\perp = \alpha p \left\|\hat{\bf x}_N\right\|_p^p>0,
\end{eqnarray*}
which contradicts the assumption that $\tilde{\bf n}_\perp \notin \pazocal{R}(A_N)$, hence $\hat{\bf x}$ is not the minimizer in such a case.

 Secondly, if we have $\tilde{\bf n}$ that is parallel to ${\bf y}$, the unique sparsest possible solution is attained only if $\tilde{\bf n}={\bf 0}$ or $\alpha$ depends on $\tilde{\bf n}$ for all $p\geq 0$ and $\alpha\geq 0$. Both events have probability 0 because ${\bf 0}$ and $\alpha$ are fixed values, which completes the proof.
\end{proof}
The result above can be extended for general penalty structure $\Gamma$ by introducing map ${\bf z}\mapsto \Gamma{\bf x}$, yielding $\left\|\Gamma{\bf x}\right\|_p^p=\left\|{\bf z}\right\|_p^p$ and system matrix $A'=A\Gamma^{-1}$.
    The theorem reveals that we cannot correctly determine the relationships among sources of different strengths when noise is present. For example, if two different sources can explain the same part of the measured noisy signal, it is impossible to tell which part belongs to which source and how much the noise smears these partial source estimates. For image reconstruction, this means we could get the image outlines correctly, but the colors may differ from the original image or blend with each other. For this reason, we introduce new terminology:

    \begin{definition}
        Sources are {\bf strongly reconstructible} if we can reconstruct them exactly. We say that sources are {\bf weakly reconstructible} if the solution $\hat{\bf z}$ has the following property
        \begin{eqnarray*}
            \left|\hat{z}_i\right|>\left|\hat{z}_j\right| \quad \textnormal{for all}\quad i\in\mathrm{supp}({\bf z}),\: j\notin \mathrm{supp}({\bf z}),
        \end{eqnarray*}
        where ${\bf z}$ is the solution for the noiseless strongly reconstructible case.
    \end{definition}

Using the concept above, we can say that the section \ref{sc:multisourceest} considered strongly reconstructible sources. Furthermore, Theorem \ref{thm:SparsestUniqMutCoh} can be extended for noisy observations as follows
\begin{theorem}
    Let ${\bf z}$ be the uniquely sparsest possible solution for $\hat{\bf y}=\mathcal{U}{\bf z}$ such that $\left\|{\bf z}\right\|_0\geq 2$. Assuming that the $m$-dimensional observations are contaminated by identically and independently distributed zero-mean Gaussian noise, the sparsest solution can be weakly reconstructed from the noisy observation with probability
    \begin{equation}
        p=\mathbb{P}\left(\tilde{X} < m(\mathrm{SNR}-1)(1-\mu^2)\right),\quad \tilde{X}\sim \chi_m^2,
    \end{equation}
    where $\mu$ is the mutual coherence, $\mathrm{SNR}$ is signal-to-noise ratio following the definition given in \cite{Calvetti2019AutommaticDepthWeighting} when considering $\hat{\bf y}$ as the signal, and $\chi_m^2$ is Chi-squared distribution of $m$ degree of freedom.
\end{theorem}
\begin{proof}
    Let us denote the identically and independently distributed Gaussian noise as $\tilde{\bf n}\sim\pazocal{N}({\bf 0},\sigma^2 \left\|\hat{\bf y}\right\|_2^2 I_m)$. Using the definition of signal-to-noise ratio (SNR), we obtain
    \begin{eqnarray*}
        \mathrm{SNR}=\frac{\mathbb{E}\left[\left\|\tilde{\bf y}\right\|_2^2\right]}{\mathbb{E}\left[\left\|\tilde{\bf n}\right\|_2^2\right]}=\frac{\left\|\hat{\bf y}\right\|_2^2+m\left\|\hat{\bf y}\right\|_2^2\sigma^2}{m\left\|\hat{\bf y}\right\|_2^2\sigma^2}=\frac{1}{m\sigma^2}+1.
    \end{eqnarray*}
    Giving
    \begin{eqnarray*}
        \sigma^2=\frac{1}{m(\mathrm{SNR}-1)}.
    \end{eqnarray*}
    Without loss of generality, we can define mutual coherence as
\begin{eqnarray*}
    \mu = \max_{i\neq k}\frac{\left|{\bf U}_i^T{\bf U}_k\right|}{\left\|{\bf U}_i\right\|_2\left\|{\bf U}_k\right\|_2}.
\end{eqnarray*}
Moreover, the notion of mutual coherence can be extended to cover systems under Gaussian noise by defining a random variable
\begin{eqnarray*}
    \tilde{M} = \frac{\left|{\bf U}_i^T\tilde {\bf y}\right|}{\left\|{\bf U}_i\right\|_2\left\|\tilde{\bf y}\right\|_2}.
\end{eqnarray*}
and
\begin{eqnarray*}
    \tilde{Y}=\frac{\tilde{\bf y}}{\left\|{\bf U}_i\right\|_2\left\|\tilde{\bf y}\right\|_2}.
\end{eqnarray*}
Since we have $\left\|{\bf U}_i\right\|_2=1$, we see that $\tilde{Y}\in\mathbb{S}^{m-1}$, which is most importantly bounded as $\mathbb{P}(\left\|\tilde{\bf y}\right\|_2=0)=0$. Now we can define a confidence radius $r_p$:
\begin{eqnarray*}
    p=\mathbb{P}\left(\sum_{i=1}^m\tilde{Z}_i^2<r_p^2\mid \tilde{Z}_i\sim \pazocal{N}(0,1),\, \forall i=1,\cdots,m\right)
\end{eqnarray*}
and note that
\begin{eqnarray*}
    \tilde{M}\leq \frac{\left|{\bf U}_i^T\hat{\bf y}+\sigma r_p\left\|\hat{\bf y}\right\|_2\right|}{\left\|\hat{\bf y}\right\|_2+\sigma r_p\left\|\hat{\bf y}\right\|_2}\leq \frac{\mu+\sigma r_p}{1+\sigma r_p}.
\end{eqnarray*}
From this, we can define {\em statistical mutual coherence} as
\begin{eqnarray*}
    \hat{\mu}=\frac{\mu +\sigma r_p}{1+\sigma r_p}.
\end{eqnarray*}
Now, by reworking Theorem \ref{thm:SparsestUniqMutCoh} for $\hat{\mu}$, we can state that the sparsest possible solution can be recovered under the noise if
\begin{eqnarray*}
    \left\|{\bf z}\right\|_0<\frac{1}{2}\left(1+\frac{1}{\hat{\mu}}\right),
\end{eqnarray*}
which does not violate the condition for the noiseless case, as $\hat{\mu}^{-1}\leq \mu^{-1}$. From the equation above, we can derive the condition 
\begin{eqnarray}\label{proof:rp}
    r_p<\frac{\mu+1-2\mu \left\|{\bf z}\right\|_0}{2\sigma(\left\|{\bf z}\right\|_0-1)},
\end{eqnarray}
where $r_p\sim\chi_m$. We also must have that the impact of mutual coherence, in addition to distortion by noise, cannot surpass the signal strength, i.e., $\mu^2+\sigma^2r_p^2< 1$. Otherwise, another column of the matrix explains the data better. However, for any $\left\|{\bf z}\right\|_0\geq 2$, we have
\begin{eqnarray*}
    \left(\frac{\mu+1-2\mu \left\|{\bf z}\right\|_0}{2(\left\|{\bf z}\right\|_0-1)}\right)^2> 1-\mu^2.
\end{eqnarray*}
    Therefore, $\mu^2+\sigma^2r_p^2< 1$ is the dominant condition from which (\ref{proof:rp}) follows.
\end{proof}

From now on, we will proceed with weakly reconstructible sources from a probabilistic point of view. The noncentral $F$-distribution has a major role in the upcoming results:

    \begin{definition}
        The noncentral $F$-distribution is continuous distribution defined in $\mathbb{R}_+$ with the following cumulative distribution:
        \begin{equation}
            F_{a,b}'(x\mid \lambda)=\sum_{j=0}^\infty \frac{(\lambda/2)^j}{j!}e^{-\lambda/2}\, \frac{B\left(\frac{ax}{b+ax};\,a/2+j,\, b/2\right)}{B(a/2+j,\, b/2)},
        \end{equation}
        where $B(\cdot,\cdot)$ is the beta function and $B(\cdot\,;\,\cdot,\,\cdot)$ is the incomplete beta function. $a$ and $b$ are the degrees of freedom and $\lambda\geq 0$ is noncentral parameter.
    \end{definition}

Using the noncentral $F$-distribution, we can conclude the following:

    \begin{lemma}
        When $\mathcal{U}$ is orthonormal, the probability of reconstruct $N$ sources weakly is
        \begin{equation}
            p\geq\mathbb{P}\left(\tilde{X}>m-Nd\right),\quad \tilde{X}\sim F'_{1,\: m-Nd}\left(\min_{i\in \pazocal{I}}\left|\mathcal{U}^T_i\hat{\bf y}\right|^2\right).
        \end{equation}
    \end{lemma}
    \begin{proof}
        Let us have the following orthonormal basis representation of the noisy observations 
        \begin{eqnarray*}
            \tilde{\bf y}=\sum_{i=1}^m\left(\bm{\psi}_{i}^T\hat{\bf y}\right)\bm{\psi}_{i}+\sum_{j=1}^m\tilde{n}_j\bm{\psi}_j,
        \end{eqnarray*}
        where $\hat{\bf y}$ is the pure observations without noise. Moreover, let $U\in \mathbb{R}^{m\times Nd}$ denote the submatrix of the system matrix $\mathcal{U}$ acting on $\mathrm{supp}(\lbrace \bm{\psi}_{i}^T\hat{\bf y}\rbrace)$. Then we assume $\mathrm{arg\min}\lbrace\left|\tilde{y}_i\right|\rbrace=k$, then the correct reconstruction by source indices is found when
        \begin{eqnarray*}
            \left|\hat{y}_k+\tilde{n}_k\right|^2>\max_{{\bf u}\in \pazocal{R}(U)^\perp}\left|{\bf u}^T\tilde{{\bf n}}\right|^2,
        \end{eqnarray*}
        which means the weakest reconstructed signals are dominated by noise-induced reconstruction.
        That can be approximated by 
        \begin{eqnarray*}
            \left\|\tilde{{\bf n}}\right\|^2_{\pazocal{R}(\mathcal{U})^\perp}\geq \max_{{\bf u}\in \pazocal{R}(U)^\perp}\left|{\bf u}^T\tilde{{\bf n}}\right|^2,
        \end{eqnarray*}
        which forms a more demanding, i.e., less probable bound:
        \begin{eqnarray*}
            \left|\hat{y}_k+\tilde{n}_k\right|^2>\left\|\tilde{{\bf n}}\right\|^2_{\pazocal{R}(\mathcal{U})^\perp},
        \end{eqnarray*}
        Noticing that
        \begin{eqnarray*}
            \tilde{X}=\frac{\left|\hat{y}_k+\tilde{n}_k\right|^2}{\left\|\tilde{{\bf n}}\right\|^2_{\pazocal{R}(\mathcal{U})^\perp}\,(m-Nd)^{-1}}\sim F'_{1,\: m-Nd}\left(\left|\hat{y}_k\right|^2\right),
        \end{eqnarray*}
        where $F'_{a,b}(\lambda)$ is the noncentral $F$-distribution. Therefore, the probability of the exact reconstruction being achievable is
        \begin{eqnarray*}
            p\geq\mathbb{P}\left(\tilde{X}>m-Nd\right).
        \end{eqnarray*}
    \end{proof}

A similar approach generalizes the result above by finding a lower bound on the norm of the correct estimate and an upper bound on the incorrect part promoted by reconstructed noise. However, the non-zero mutual coherence poses more restrictions. For instance, a non-degeneracy condition, which says that there is a strict lower bound on how much of a certain vector lies in the range of a matrix, is required so that the reconstructed noise does not blow up.

\begin{lemma}\label{lem:signalnuisancebounds}
    Let the observations be generated as $\tilde{\bf y}= A{\bf z}_0+\tilde{\bf n}$, where $A$ is orthonormal, $\left\|{\bf z}\right\|_0=N$ as the sparsest possible solution, and $\tilde{\bf n}\notin \pazocal{R}(A)$ and non-degeneracy condition $\left\|P_B\tilde{\bf n}\right\|_2\geq c\left\|\tilde{\bf n}\right\|_2$ holds for $c>0$, where $P_B$ is orthogonal projection onto image of matrix $B\in\mathbb{R}^{m\times (m-N)}$. Assuming further that the system matrix $\begin{bmatrix}
        A&B
    \end{bmatrix}$ has mutual coherence $\mu$ with a property $(m-N-1)\mu <1$, a signal estimation has a lower bound
    \begin{equation}
        \left\|\hat{\bf z}\right\|_2\geq  \left\|{\bf z}_0\right\|_2-\mu \sqrt{N(m-N)}\frac{\sqrt{1+(m-N-1)\mu}}{c(1-(m-N-1)\mu)}\left\|\tilde{\bf n}\right\|_2.
    \end{equation}
    and a nuisance estimate $\hat{\bf x}$ generated by reconstructing noise as a signal has an upper bound
    \begin{equation}
        \left\|\hat{\bf x}\right\|_2\leq \frac{\sqrt{1+(m-N-1)\mu}}{c(1-(m-N-1)\mu)}\left\|\tilde{\bf n}\right\|_2
    \end{equation}
\end{lemma}

\begin{proof}
Let ${\bf y}=A{\bf z}_0$, where $\left\|{\bf z}_0\right\|_0=N$ and it is the unique sparsest solution for the given ${\bf y}$ and $A\in\mathbb{R}^{m\times N}$. Since $\tilde{\bf n}\in\mathbb{R}^m$ is orthogonal with ${\bf y}$, we can write the minimization problem as
\begin{eqnarray*}
    F({\bf z},{\bf x})=\left\|{\bf y}-A{\bf z}-P_{\bf y}B{\bf x}\right\|_2^2+\left\|\tilde{\bf n}-P_{\bf n}B{\bf x}\right\|_2^2,
\end{eqnarray*}
where $P_{\bf y}={\bf y}{\bf y}^T/\left\|{\bf y}\right\|_2^2$ and $P_{\bf n}=\tilde{\bf n}\tilde{\bf n}^T/\left\|\tilde{\bf n}\right\|_2^2$ are orthogonal projections. Moreover, the second projection is valid because $\mathbb{P}(\tilde{\bf n}={\bf 0})=0$.

To minimize the square of the norm on the left-hand side, we can choose $\hat{\bf z}$ such that
\begin{eqnarray*}
    \hat{\bf z}({\bf x})={\bf z}_0-{\bf z}_0{\bf y}^TB{\bf x},
\end{eqnarray*}
for any arbitrary ${\bf x}$. Lower bound for $\hat{\bf z}$ can be obtained by
\begin{eqnarray*}
    \left\|\hat{\bf z}\right\|_2\geq \left\|{\bf z}_0\right\|_2-\left\|P_{\bf y}B{\bf x}\right\|_2\geq \left\|{\bf z}_0\right\|_2-\mu \sqrt{N(m-N)}\left\|\hat{\bf x}\right\|_2
\end{eqnarray*}

By Gershgorin's circle theorem \cite{Gershgorin1931}, every eigenvalue of $B^TB$ lies in the interval 
\begin{eqnarray*}
    \left[1-(m-N-1)\mu,\; 1+(m-N-1)\mu\right],
\end{eqnarray*}
The previous result holds due to the fact that the eigenvalues of $B^TB$ must lie on the $\left[B^TB\right]_{ii}$-centric disk (that is at most 1-centric), where the radius is
\begin{equation}
    R_i=\sum_{j\neq i}\left|\left[B^TB\right]_{ij}\right|,
\end{equation}

Due to the non-degeneracy, we have a vector ${\bf v}$ such that $B\tilde{\bf v}=P_B\tilde{\bf n}$, hence
\begin{eqnarray*}
    \left\|B^T\tilde{\bf n}\right\|_2=\left\|B^TB\tilde{\bf v}\right\|_2\geq (1-(m-N-1)\mu)\left\|\tilde{\bf v}\right\|_2,
\end{eqnarray*}
while
\begin{eqnarray*}
    \left\|P_B\tilde{\bf n}\right\|_2=\left\|B\tilde{\bf v}\right\|_2\leq \sqrt{1+(m-N-1)\mu}\left\|\tilde{\bf v}\right\|_2
\end{eqnarray*}
Yielding
\begin{eqnarray*}
    \left\|\hat{\bf x}\right\|_2\leq \frac{\sqrt{1+(m-N-1)\mu}}{c(1-(m-N-1)\mu)}\left\|\tilde{\bf n}\right\|_2
\end{eqnarray*}
and similarly
\begin{eqnarray*}
    \left\|\hat{\bf z}\right\|_2\geq  \left\|{\bf z}_0\right\|_2-\mu \sqrt{N(m-N)}\frac{\sqrt{1+(m-N-1)\mu}}{c(1-(m-N-1)\mu)}\left\|\tilde{\bf n}\right\|_2.
\end{eqnarray*}
\end{proof}

And these bounds provide a lower bound for the probability of weakly reconstructing the source that is not necessarily the sparsest, as the final theoretical conclusion:

    \begin{theorem}\label{thm:GMMexactprob}
    Let the observations be constructed from $N$ sources, where the separate observations for individual sources are orthogonal. Then, the solution for ${\bf y}=\mathcal{U}{\bf z}_0$ can be weakly reconstructed from ${\bf y}+\tilde{\bf n}=\mathcal{U}{\bf z}$ with probability
    \begin{equation}
    \begin{split}
        p&\geq \mathbb{P}\left(\tilde{X}>\left(1+\mu(\mathcal{U}) \sqrt{Nd(m-Nd)}\right)^2\frac{(m-Nd)(1+(m-Nd-1)\mu(\mathcal{U}))}{c^2(1-(m-Nd-1)\mu(\mathcal{U}))^2}\right),\\
        \tilde{X}&\sim F'_{Nd,\: m-Nd}\left(\min_{i\in 1,\dots Nd}\left|y_i\right|^2\right),
    \end{split}
    \end{equation}
    where $\tilde{\bf n}$ is Gaussian distributed, and if the conditions from Lemma \ref{lem:signalnuisancebounds} holds.
    \end{theorem}
\begin{proof}
    Let us define the signal space $\pazocal{S}_s\subset \pazocal{R}(\mathcal{U})$ and its complement $\pazocal{S}_\perp$ such $\pazocal{S}_s\oplus \pazocal{S}_\perp\subset \pazocal{R}(\mathcal{U})$.  Now, recalling Lemma \ref{lem:signalnuisancebounds}, we get a "worst-case condition" for the weakly reconstructible as
    \begin{eqnarray*}
    \left|\hat{z}_k\right|\geq  \left|z_{0;k}\right|-\mu \sqrt{N(m-Nd)}\lambda\left\|\tilde{\bf n}\right\|_2\geq \lambda\left\|\tilde{\bf n}\right\|_2\geq \left\|\hat{\bf x}\right\|_2,
    \end{eqnarray*}
    where $\lambda=\sqrt{1+(m-Nd-1)\mu}/c(1-(m-Nd-1)\mu)$. This ensures that the weakest possible estimate of the actual source is larger than the strongest reconstructed pure noise, as $\max|\hat{x}_i|\leq \left\|\hat{\bf x}\right\|_2$. This leads to a  requirement
    \begin{eqnarray*}
    \frac{\left|{\bf z}_{0,k}\right|}{\left\|\tilde{\bf n}\right\|_2}=\frac{\left|\tilde{y}_k\right|}{\left\|\tilde{\bf n}\right\|_2}>\left(1+\mu \sqrt{Nd(m-Nd)}\right)\lambda.
\end{eqnarray*}
    Modifying this to the following form
    \begin{eqnarray*}
        \tilde{X}=(m-N)\frac{\left|\tilde{y}_k\right|^2}{\left\|\tilde{\bf n}\right\|_2^2}>\left(1+\mu \sqrt{Nd(m-Nd)}\right)^2(m-Nd)\lambda^2
    \end{eqnarray*}.
    Hence, we know that $\tilde{X}$ is $F_{Nd,m-Nd}'\left(\left|y_k\right|^2\right)$ distributed. 
    
\end{proof}

In some applications, the observation noise is assumed to be independent and identically distributed, and the signal-to-noise ratio is known. In such a case, we can obtain 
\begin{eqnarray}\label{eq:SNR}
    \mathrm{SNR}=\frac{\mathbb{E}\left[\left\|\tilde{\bf y}\right\|_2^2\right]}{\mathbb{E}\left[\left\|\tilde{\bf n}\right\|_2^2\right]}=\frac{\left\|{\bf y}\right\|_2^2+m\sigma}{m\sigma}=\frac{\left\|{\bf y}\right\|_2^2}{m\sigma}+1,
\end{eqnarray}
and hence there exists some probability, lower bounded by Theorem \ref{thm:GMMexactprob}, where the random variable can be written in terms of SNR:
\begin{equation}
    \tilde{X}\sim F'_{d,\: m-d}\left(m(\mathrm{SNR}-1)\right).
\end{equation}
From the result, we note that, contrary to the results above, the noncentral parameter does not depend on the weakest contributor to the data. This indicates that localization accuracy is not depth-dependent in this case.

\section{Methods}
In this section, we introduce the new methods Unbiased Gaussian Estimation (UGE) and the Exact-Coloring Bayesian method (EC Bayesian) used in the following experiments. Both methods are motivated by the theory presented in previous sections. The purpose of these methods is to showcase the advantages of unbiased approaches.

\subsection{sLORETA and its $d$-dimensional extensions}
Standardized low-resolution electromagnetic brain tomography (sLORETA) is presented in Eqs. (\ref{eq:sLOR1}) and (\ref{eq:sLOR2}) as a two-step method that aims first to find the correct system matrix column, i.e., the correct source, by forming a Z-score field, and then to fit the source. This was originally written in the 1-dimensional case, but an extension was also provided in \cite{Pascual2007sqrtm}. The method can be extended to dimension $d$ by considering the $d\times d$ diagonal blocks of the model covariance matrix
\begin{equation}
    \Sigma=LL^T+C\in \mathbb{R}^{nd\times nd}.
\end{equation}
By denoting $k$th block as $\Sigma_k$, Eq. (\ref{eq:sLOR1}) can be reformulated as
\begin{equation}\label{eq:sLORa}
    \lbrace k_1,\dots, k_d\rbrace=\argmax{\frac{L_i^T\Sigma_i^{-1}{\bf y}}{\sqrt{L_i^T\Sigma^{-1}_iL_i}}}{i=1,\cdots,n}.
\end{equation}

\subsection{Unbiased Gaussian Estimation}
We denote now the observation model in an unbiased setting as
\begin{equation}
    \tilde{\bf y}=\mathcal{U}\tilde{\bf z}+\tilde{{\bf n}},
\end{equation}
where $\left\|\tilde{\bf z}\right\|_0=Nd$ and $\tilde{\bf n}$ are drawn from continuous distributions.

If we consider this in terms of index sets, we can search for the solution from the source index sets $\lbrace\pazocal{I}_k\colon \left|\pazocal{I}_k\right|=N\rbrace$, assuming the number of sources ($N$) is given a priori. By denoting $U_k\in\mathbb{R}^{m\times Nd}$ as the sub-system corresponding to sources represented by $\pazocal{I}_k$ that is aimed to solve $\hat{\bf y}=U_k{\bf z}_k$. Then, we can formulate a corresponding Bayesian model for $N$- source estimation as
\begin{equation}\label{eq:BayesModel}
    p(\hat{\bf y}\mid {\bf z})=\sum_{j=1}^{\binom{n}{N}}\pi_{j}\pazocal{N}\left( U_{j}{\bf z}_j,\hat{C}\right),
\end{equation}
where $\hat{C}=\Sigma^{-1/2}C\Sigma^{-1/2}$. Since the Gaussian prior model is encoded into $\hat{\bf y}$ and $\mathcal{U}$, we can solve the maximum a posteriori (MAP) as we would solve maximum likelihood estimation, which is easy to check by solving the $\Sigma$-standardized observation model. Hence, by denoting
\begin{equation}\label{eq:GMMh}
    h_{k}=\frac{\pi_{k}p(\hat{\bf y}\mid U_{k}{\bf z}_{k},\hat{C})}{\sum_{\pazocal{I}\colon \left|\pazocal{I}\right|=N}\pi_{\pazocal{I}}p(\hat{\bf y}\mid U_{\pazocal{I}}{\bf z}_{\pazocal{I}},\hat{C})},
\end{equation}
 we can write the posterior estimate as
\begin{equation}\label{eq:GMMest}
    \hat{\bf z}=\sum_{\pazocal{I}\colon \left|\pazocal{I}\right|=N}\frac{h_{\pazocal{I}}}{\sum_{\pazocal{I}\colon \left|\pazocal{I}\right|=N} h_{\pazocal{I}}}E_{nd;\pazocal{I}}(\hat{C}^{-1/2}U_{\pazocal{I}})^\dagger \hat{C}^{-1/2}\hat{\bf y},
\end{equation}
where $E_{nd;\pazocal{I}}\in\mathbb{R}^{nd\times Nd}$ projection matrix from sub-index set $\pazocal{I}$ to the full $nd$-dimensional parameter space to the corresponding element slots. Within the model, we promote unbiased estimation by setting all $\pi_{k}$ equal. Furthermore, since in the model $(\hat{C}^{-1/2}\mathcal{U},\, \hat{C}^{-1/2}\hat{\bf y})$ both objects are standardized by noise, we can use the previously introduced change of variables by redefining the observations and source-wise computing singular value decomposition of $C^{-1/2}L$ and using the left singular matrix as the representation of the forward model.

So, the Unbiased Gaussian Estimation is a Gaussian mixture model-based approach presented in Eq. \ref{eq:BayesModel}. The method is based on the idea that a sub-system matrix, denoted $U_\pazocal{I}$, of the full system matrix $\mathcal{U}\in\mathbb{R}^{m\times nd}$ explains the data. Technically, the proper solution can be found by computing the mixture model of all possible sub-system matrix combinations. However, this is computationally too cumbersome, so we assume the number of sources is known. Moreover, in the implementation, we used a Monte Carlo sampling strategy to generate 1,000 random samples, i.e., index sets $\hat{\pazocal{I}}$ such that the size of the index set is $N\cdot d$ for $N$ sources of dimension $d$.

The following presents the pseudocode for the algorithm:

\begin{algorithm}
\caption{Pseudocode for UGE algorithm}\label{algo:UGE}
\begin{algorithmic}[1]
\State \textbf{Input:} $K$ number of random sets, assumed number of sources $N$, system matrix $L\in\mathbb{R}^{m\times nd}$, observations ${\bf y}$, prior source covariance $P$, and observation noise covariance $C$
    \State Compute $\Sigma=LPL^T+C$, $\hat{L}=\Sigma^{-1/2}LP$, $\hat{\bf y}=\Sigma^{-1/2}{\bf y}$, and $\hat{C}=\Sigma^{-1/2}C\Sigma^{-1/2}$.
    \State Compute compact SVD of sub-system matrices $U_i\in\mathbb{R}^{m\times d}$ and $V_i\in\mathbb{R}^{d\times d}$ for $i=1,\cdots ,n$ using $\hat{L}$.
    \State Create $\mathcal{U}=\begin{bmatrix}
        U_1V_1^T&\cdots& U_nV_n^T
    \end{bmatrix}$.
    \State Generate random source index sets $\lbrace \pazocal{I}_1,\cdots \pazocal{I}_K\rbrace$ of size $N$ corresponding a sub system matrix $\mathcal{U}_{k}\in\mathbb{R}^{m\times Nd}$ (one source has $d$-dimensional representation) at uniformly random for $k=1,\dots, K$.
    \For{$k=1,\dots , K$}
    \State Compute $\hat{\bf z}_k=\mathcal{U}_k^\dagger \hat{\bf y}$.
    \State Compute $p(\hat{\bf y}\mid \hat{\bf z}_k)=\pazocal{N}(\hat{\bf y}\mid \mathcal{U}_k\hat{\bf z}_k, \hat{C})$.
    \EndFor
    \State Compute full estimate using Eqs. (\ref{eq:GMMh}) and (\ref{eq:GMMest}).
\end{algorithmic}
\end{algorithm}

\subsection{Exact-Coloring Bayesian method}
The Exact-Coloring Bayesian (EC Bayesian) method modifies the previously introduced UGE method and is tailored to experiments with the Shepp-Logan phantom, where the "color" is the exact magnitude of the source. The known magnitudes allow us to search a discrete set of feasible solutions in a similar manner to how we would seek the sparsest possible solution. The pseudocode for the EC Bayesian reads as follows:

\newpage

\begin{algorithm}
\caption{Pseudocode for EC Bayesian algorithm}\label{algo:ECBayes}
\begin{algorithmic}[1]
\State \textbf{Input:} $K$ number of random sets, assumed number of sources $N$, system matrix $L\in\mathbb{R}^{m\times nd}$, observations ${\bf y}$, prior source covariance $P$, and observation noise covariance $C$, color/amplitude values ${\bf v}\in \mathbb{R}^s$
    \State Compute the unbiased system matrix and auxiliary variables according to steps 2--4 of Algorithm \ref{algo:UGE}.
    \State Generate random source index sets $\lbrace \pazocal{I}_1,\dots, \pazocal{I}_K\rbrace$ of size $N$ corresponding a sub system matrix $\mathcal{U}_{k}\in\mathbb{R}^{m\times Nd}$ (one source has $d$-dimensional representation) for $k=1,\dots, K$.
    \For{$k=1,\dots , K$}
    \State Compute ${\bf z}_k=\mathcal{U}_k^\dagger \hat{\bf y}$.
    \State Compute random seed $K$-means for amplitudes of ${\bf z}_k$, where we seek $s$ clusters. 
    \State Assign values from ${\bf v}$ for each cluster of ${\bf z}_k$ such that $p(\hat{\bf y}\mid \hat{\bf z}_k)=\pazocal{N}(\hat{\bf y}\mid \mathcal{U}_k \hat{\bf z},\hat{C})$ is maximized, where $\hat{\bf z}_k$ amplitudes belong to the set $\lbrace v_s,\dots , v_s\rbrace$.
    \EndFor
    \State Select such $\hat{\bf z}_k$ as the estimate that has maximum $p(\hat{\bf y}\mid \hat{\bf z}_k)$.
\end{algorithmic}
\end{algorithm}

Providing the set of exact values a non-zero parameter can be, i.e., $z_k\in \pazocal{C}$, would increase the noise robustness of the method. Denoting $s=\spark(\mathcal{U})$, we can obtain that for some data ${\bf y}$:
\begin{equation}
    {\bf y}=\sum_{j=1}^k{\bf U}_jz_j=-\sum_{j=k+1}^s{\bf U}_jz_j,\quad \left\|{\bf U}_i\right\|_2=1,
\end{equation}
where the unique sparsest possible solution is $k$-sparse. In such a case, $s\leq 2k$ leads to contradiction either for uniqueness or the $k$ being smallest sparsity. For $s>2k$, we can form a lower bound for the error:
\begin{equation}
    \left\|{\bf y}-\sum_{j=k+1}^{2k}{\bf U}_j(-z_j)\right\|_2^2\geq(1-\mu(s-2k-1))(s-2k)\min_{c_j\in\pazocal{C}} \left|c_j\right|^2.
\end{equation}
Therefore, if the distortion caused by the noise is smaller than the given bound, that is $\left\|\tilde{\bf n}\right\|_2^2<(1-\mu(s-2k-1))(s-2k)\min_{c_j\in\pazocal{C}} \left|c_j\right|^2$, the correct solution is guaranteed to be recovered under the noise.

\section{Experiments}\label{sc:experiments}
\subsection{Experiment set 1}
Our first experiment examines the consequences of the mathematical lemmas and theorems from section \ref{sc:NoisyEstimation} in a setup often used in compressed sensing studies with tomographic Fourier sampling. In experiment set 1, we follow the experiment of Candes {\em et al.} \cite{CandesEmmanuel2006} to reconstruct a $256\times 256$ Shepp-Logan phantom from 22 tomographic line samples. The experimental setup fits this article's themes, not only because of the exact reconstruction topic. But also because the image "source" is not sparse, as the image content occupies about 42 \% of the black canvas; therefore, the phantom image cannot be reconstructed using solely standardized methods specialized for locating a sparse source.

Since the noiseless case is well established, we conduct Experiment 1.A., where the image is corrupted by 30 \% of independent and identically distributed (i.i.d.) Gaussian noise, and Experiment 1.B, where the data is contaminated by 30 \% of spatially correlated Gaussian noise. Due to the computational limitations regarding full covariance matrices, the phantom in Experiment 1.B is the size of $128\times 128$. As 22-line samples violate the unique reconstruction condition presented in Theorem \ref{thm:exactreclimit}, we conduct Experiment 1.C. with 10 \% i.i.d. Gaussian noise and 320 lines, which is the minimum number that fulfills the condition. In addition, Experiment 1.B is redone with 320 line samples. It has been demonstrated by \cite{ChartrandR2008randomvstomographicsampling} that tomographic line sampling is preferable to random sampling for total variation-minimizing methods, so we decided to stick with this sampling approach.

To avoid the large memory consumption caused by the construction of the forward model of the restricted Fourier transform with $256^2$ complex parameters, we compute the Fourier transform and its inverse directly while utilizing Split Bregman iterations \cite{Goldstein2009SplitBregman}. The Gaussian mixture approach was modified for this experiment to use the Gaussian prior $p({\bf z})\propto \exp(-\left\|D{\bf z}\right\|_2^2)$, where $D$ is the cyclic image derivative. This approach has been shown to be highly efficient for image reconstruction from sparse data \cite{CandesEmmanuel2006}. As the prior requires a rectangular spatial structure in the "spatial" image space, we use random cyclic $250\times 250$ and $96\times 96$ spatial sets in Experiment 1.B as the candidate space for the image content. 1,000 random spatial samples were used to reconstruct the image. Because the phantom image contains only black, white, and five shades of gray, we demonstrate a reconstruction approach that uses prior knowledge of the relative differences among the discrete source amplitudes, which is highly relevant to the theoretical results of this paper. Applying the said lemma and Theorem \ref{thm:exactreclimit}, we can derive the expectation that the method reconstructs the phantom image nearly perfectly with overwhelmingly high probability in Experiment 1.C with 10 \% noise. As discussed in the previous section, exact amplitudes are not expected unless the forward model forms an orthogonal basis. The methods are compared with the Spatial Total Variation Split Bregman (STVSB) method \cite{Montesinos2014STVSB}. 

Because restricted Fourier transforms meet the criterion for an unbiased forward matrix, we can expect the reconstruction to be highly noise-robust; as such, the observation noise was set particularly high to demonstrate the effects. Especially the spatially correlated noise in Experiment 1.B is disadvantageous to the selected prior- and total-variation-based approaches in general. Given the Bayesian framework and appropriate likelihood modeling, the accuracy differences among the methods considered in this paper should not be striking, while STVSB could perform worse in Experiment 1.B. 

\subsection{Experiment set 2}
In the second experimental set, we focus on source estimation for a homogeneous conductivity disk. Namely, we demonstrate practical consequences of Theorem \ref{thm:GMMexactprob} for reconstructions of sources at different depths and how known signal-to-noise ratios impact the probability of perfectly reconstructing sources with a different number of sensors. Moreover, we examine one-source estimation and a couple of source scenarios: simultaneous superficial (near-field) sources and a combination of superficial and far-field sources. The sources are placed in such a way that the source is a deep source, which is hard to estimate; the second case has two sources that produce orthogonal contributions in the observation space, and the third setup contains a cortical and a deep source with similar contributions to observations, and as such is the hardest setup of these three. The observations were corrupted by five percent of additive Gaussian noise. The purpose of this experiment set is to extend the results obtained with the same model in \cite{Lahtinen2024}, where a lower bound for exact reconstruction was introduced as evidence about the possible reasons for the high noise robustness of unbiased methods. Sources are estimated using Minimum Norm Estimate (MNE) \cite{HamalainenMNE}, Unbiased Gaussian Estimate (UGE) from Theorem \ref{thm:GMMexactprob}, a classical standardized low-resolution electromagnetic brain tomography (sLORETA) \cite{PascualMarqui2002} method that standardizes each parameter individually; it is a 1-parameter model and its extension, in which we associate a 2-dimensional orientation with each source. We call this extension sLORETA2D. In a theoretical sense, there is no "correct" variant of sLORETA. The sLORETA variants just approach the inversion problem differently.

\subsection{Experiment set 3}
The third and final set of experiments was conducted using a realistic, volumetric, multicompartment model of the human brain. The MRI used was obtained from the Ernie subject provided by SimNIBS\footnote{\url{https://simnibs.github.io/simnibs/build/html/index.html}} 4 software \cite{PUONTI2020117044}. The mesh has 743,575 tetrahedral elements joined in 136,868 nodes. Four tissue compartments were used (scalp, skull, cerebrospinal fluid, and brain), and 76 electrodes were placed on the scalp layer according to the international 10-10 system. 

The lead field matrix used in this study was constructed using custom software that employed the Finite Element Method with linear basis functions, as in \cite{Wolters2004}. The tissue electric conductivity values were 0.43 S/m for the scalp, 0.0103 S/m for the skull, 1.79 S/m for cerebrospinal fluid, and 0.33 S/m for the brain (gray matter and white matter) \cite{ram06}. The realistic geometries are extracted from MRI images, and a linear forward model is built by approximately solving quasi-static Maxwell's equations using the finite element method. The problem is to solve for 3-dimensional dipolar moments, i.e., basis coefficients for basis functions that represent the dipolar term causing the electric potential field inside the head. The observations are generated synthetically using a separate forward model with different source locations from those used in the inversion to avoid inverse crime. Both models have approximately 10,000 fixed source locations.

The true source placement follows the idea of Experiment 2.B, where one source is in the near-field, and the other is a far-field source. We qualitatively assess the reconstructions from sLORETA, sLORETA3D (3-dimensionally extended sLORETA), and UGE by visualizing the Z-score distributions on MRI slices for each true source separately. Moreover, we compute the 2-cluster dipole localization error (DLE), where two clusters were identified from the spatial Z-score distributions and use the cluster means as estimates of the source location. To validate the spatial spread of distributions, we use the Earth Mover's Distance (EMD).

\section{Results}\label{sc:results}
\subsection{Experiment set 1}
The results of Shepp-Logan phantom image reconstruction from 22 radial lines are presented in Figure \ref{fig:phantom22}. Reconstructed images with 30 \% i.i.d. noise are recognizable. However, the amplitude, i.e., the shade, of the darkest gray is too weak and dark for STVSB. UGE has similar problems, but to a lesser extent; while the method that uses shade ratios gets the darkest gray right, it misses all the lighter gray features. From the error highlight, we see that both STVSB and UGE reconstruct a notable amount of noise. Based on the reconstructions, STVSB seems to maintain sharper edges around the white ellipse than UGE, which produces blurrier reconstructions due to its Gaussian nature. 

\begin{figure}[h!]
    \centering
    \includegraphics[width=0.86\linewidth]{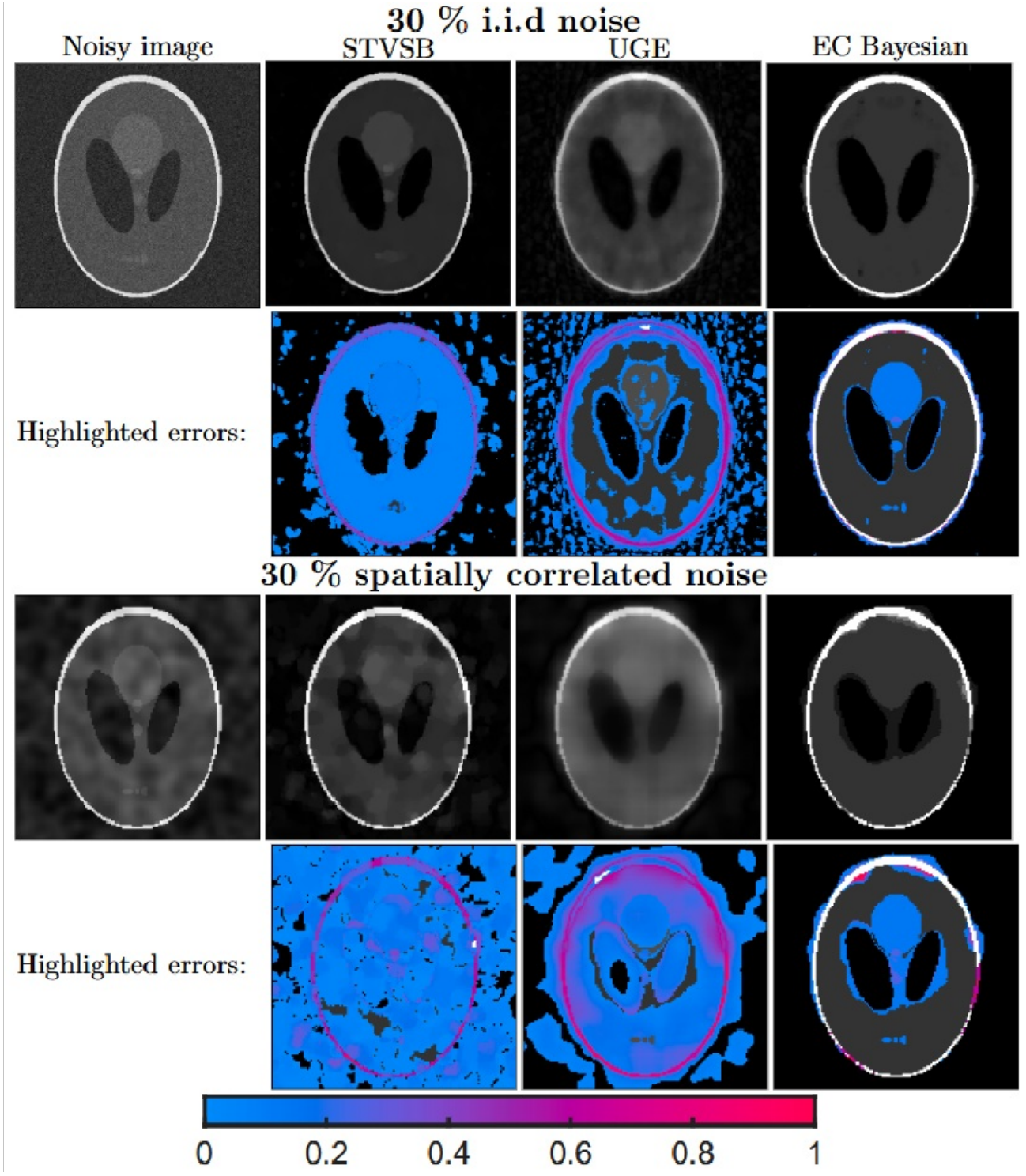}   
    \caption{Reconstructions of noisy Shepp-Logan phantoms from 22 radial samples and the highlights of the errors in the nebula color map plotted over the clean phantom. Reconstructions are computed with Spatial Total Variation Split Bregman (STVSB) with 100 iterations, Unbiased Gaussian Estimate (UGE), and the Exact Coloring Bayesian method EC Bayesian, which has the prior information about the ratio of all seven colors appearing in the true picture. UGE and EC Bayesian were computed using 1,000 samples. Color values range from 0 to 1. The error metric is the pixel-wise absolute error between the complex-valued estimate and the true image, where values range from 0 to 1 (black to white).}
    \label{fig:phantom22}
\end{figure}

On the one hand, maintaining drastic color changes in the image content is desirable; on the other hand, it makes the method prone to introducing noise into the reconstruction. This drawback is more pronounced when the noise is correlated, indicating that STVSB reconstruction essentially reconstructs the noisy image. In contrast, Bayesian methods that use a better noise model can reduce it. The graphs in Figure \ref{fig:phantom_relativ_err_graphs1} show the relative error of STVSB for the default 100 iterations and the two other methods up to 500 sampled partial reconstructions. STVSB error stabilizes at 0.3, EC Bayesian achieves a lower error around 500 samples, and UGE error is slightly below 0.4. The sampling-based methods stabilized around 500 samples.

The results with 320 radial sampling lines in Figure \ref{fig:phantom320} are computed up to 1,000 samples for UGE and EC Bayesian. The result shows nearly perfect reconstruction with EC Bayesian, with errors only at the black inner boundaries. Error is higher with 30 \% of the noise, but still less than 0.1 (Figure \ref{fig:phantom_relativ_err_graphs2}). The worst EC Bayesian reconstruction is obtained with spatially correlated 30 \% noise. Moreover, compared to the 22 radial line case, STVSB reconstruction quality remains about the same, unlike UGE, which provides significant improvement. Based on the relative error graphs, both Bayesian methods eventually outperform STVSB. With 30\% noise, UGE and EC Bayesian stabilized around 500 samples. However, with 10\% i.i.d. noise, the UGE and EC Bayesian errors did not stabilize during those 1,000 samples, so the errors could end up lower after a couple of thousand iterations.

\begin{figure}
    \centering
    \begin{minipage}{0.4\linewidth}
    \centering
        \footnotesize{30 \% i.i.d.}
        
        \includegraphics[width=0.98\linewidth]{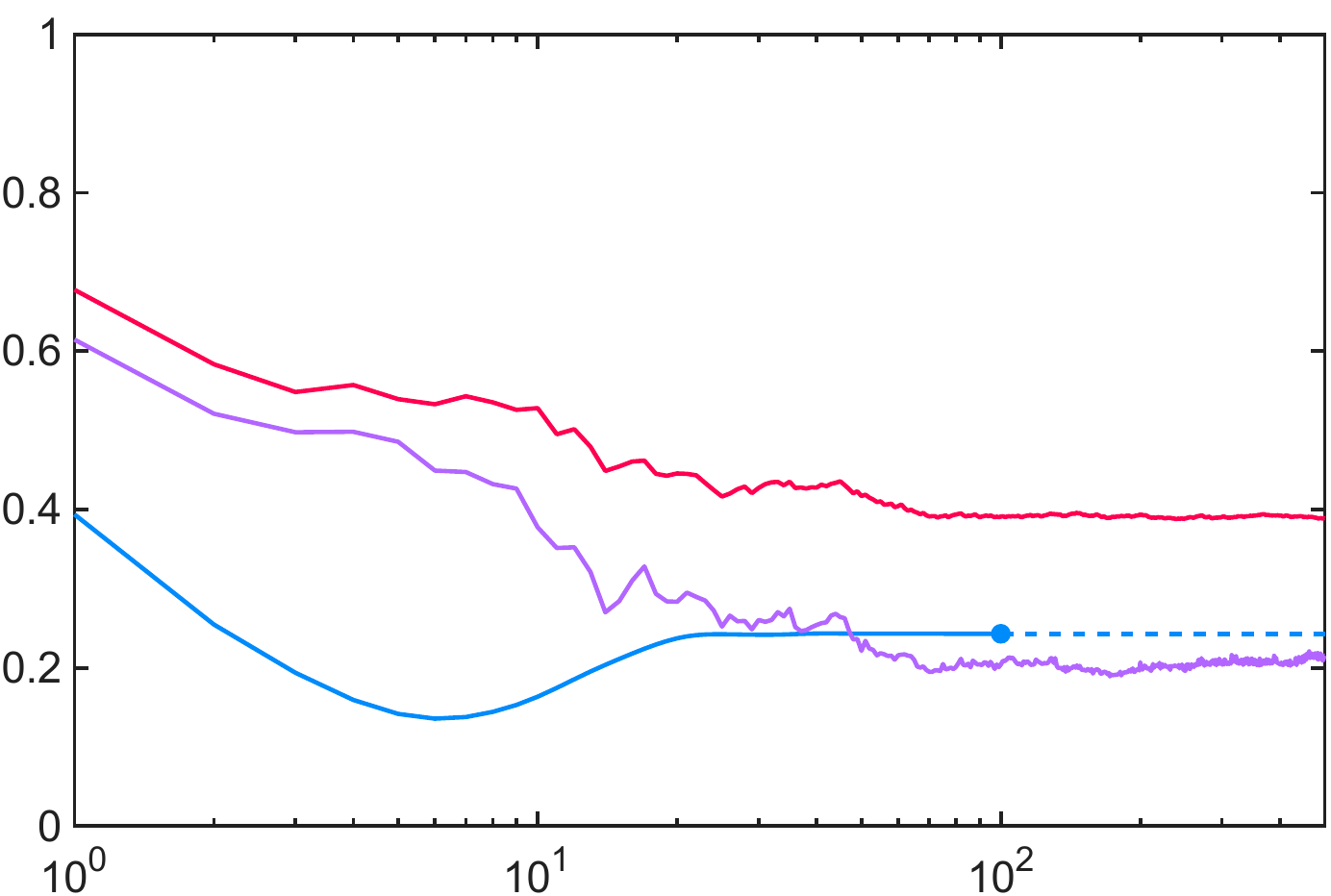}
    \end{minipage}\begin{minipage}{0.4\linewidth}
    \centering
        \footnotesize{30 \% spatially correlated}
        \includegraphics[width=0.98\linewidth]{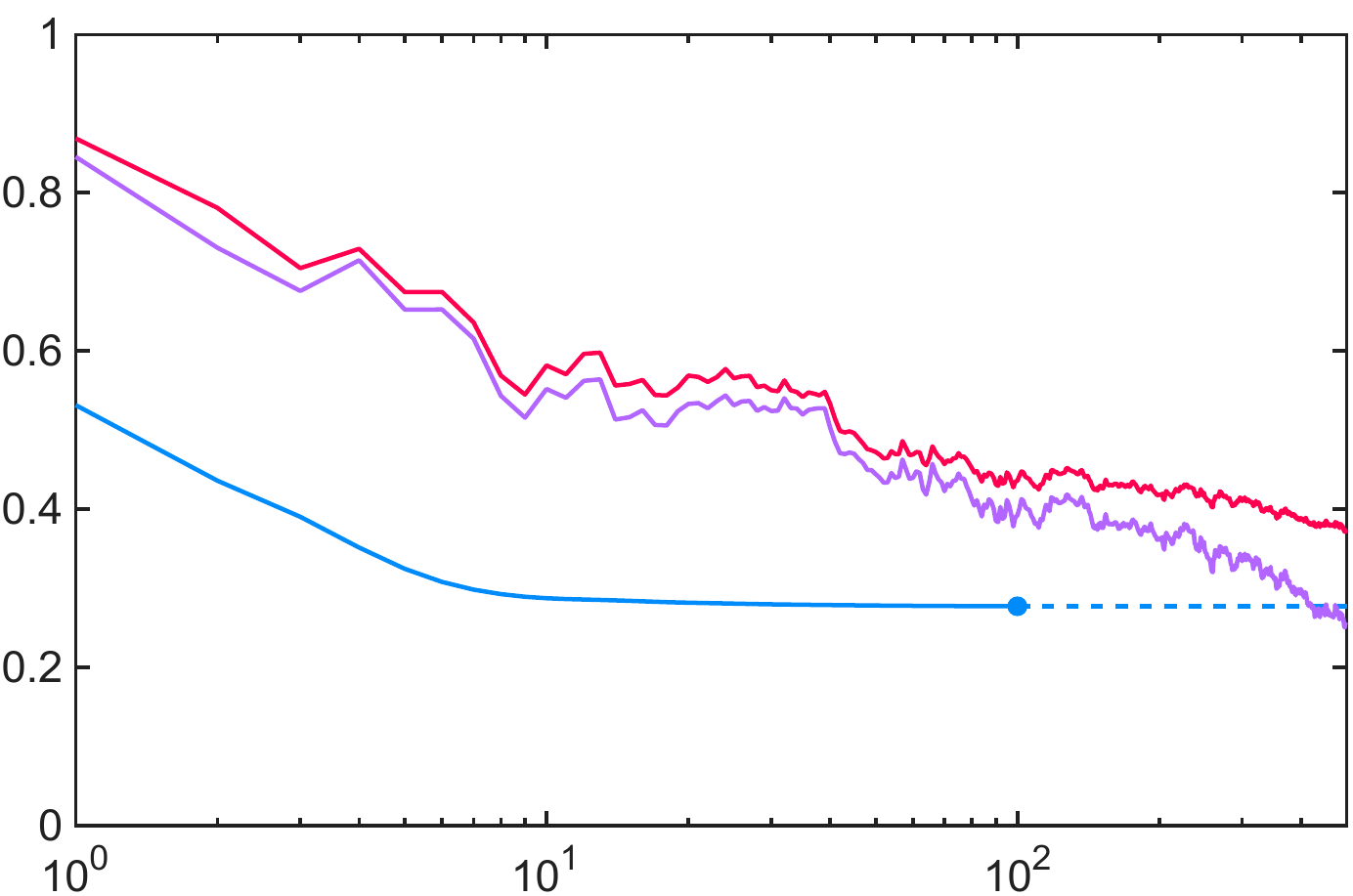}
    \end{minipage}
    \caption{Relative errors of the method for each foremost iteration. Iterations are presented on a logarithmic x-axis and errors on a y-axis. STVSB (light blue) is computed over 100 split Bregman iterations. UGE (red) and EC Bayesian (violet) are computed over 500 samples.}
    \label{fig:phantom_relativ_err_graphs1}
\end{figure}

\begin{table}[h]
    \centering
    \begin{tabular}{r|c|c|c} \hline
        Method & STVSB & UGE & EC Bayesian\\
        Time & 13.75 \unit{\second} & 9 \unit{\minute} 51.51 \unit{\second} & 28 \unit{\minute} 26.70 \unit{\second}
    \end{tabular}
    \caption{Computation times of each algorithm used in Experiment 1. Spatial Total Variation Split Bregman (STVSB) was run with 100 iterations, and Unbiased Gaussian Estimate (UGE) and Exact Coloring Bayesian (EC Bayesian) were computed over 1,000 spatial samples.}
    \label{tab:computationtimes}
\end{table}

\clearpage

\begin{figure}
    \centering
    \includegraphics[width=0.8\linewidth]{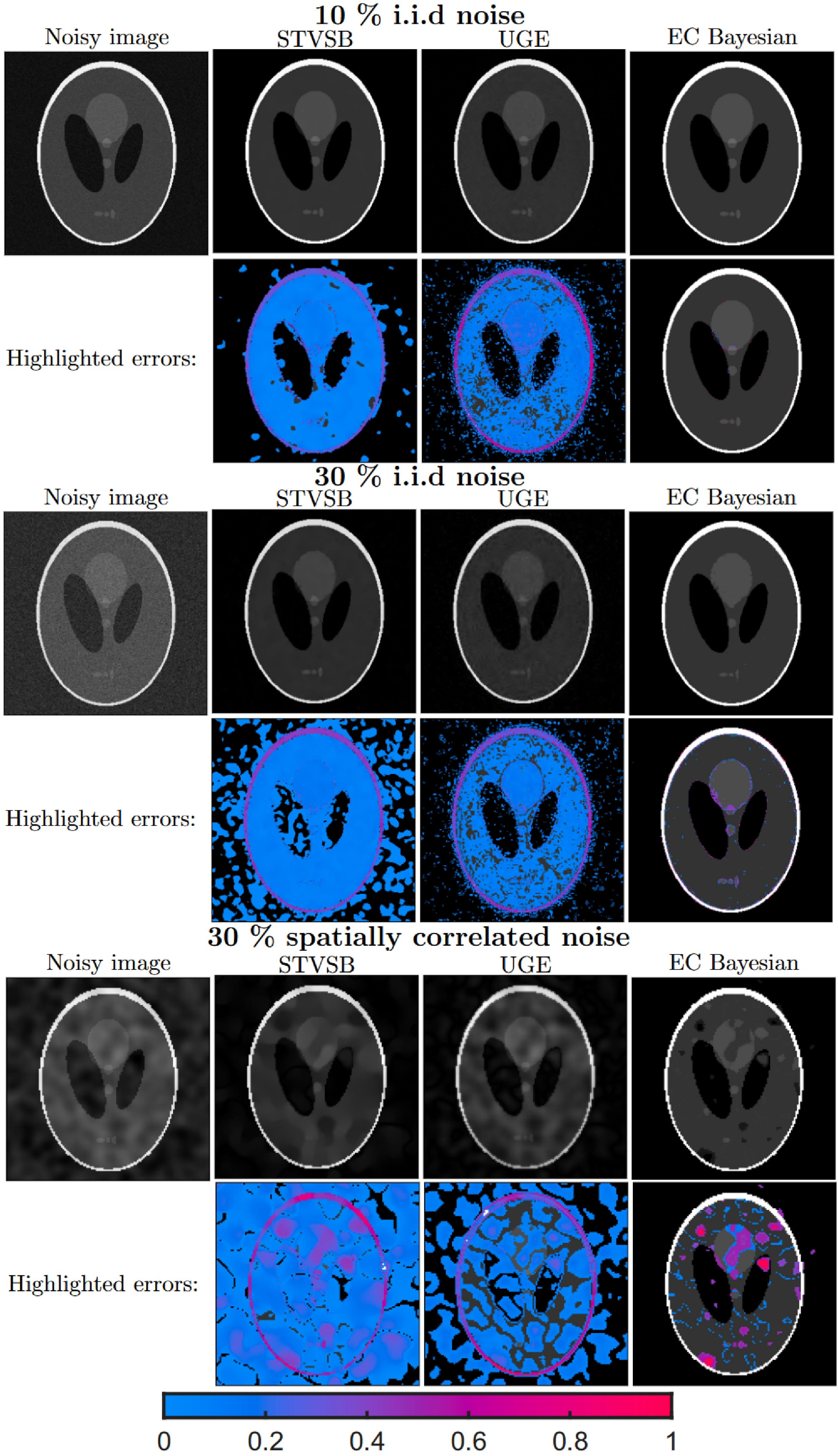}
    \caption{The figure follows the same logic as Figure \ref{fig:phantom22}, but now we have 320 line samples, which is the minimum amount to reconstruct the noiseless phantom exactly.}
    \label{fig:phantom320}
\end{figure}

\begin{figure}
    \centering
    \begin{minipage}{0.3\linewidth}
    \centering
        \footnotesize{10 \% i.i.d.}
        
        \includegraphics[width=0.98\linewidth]{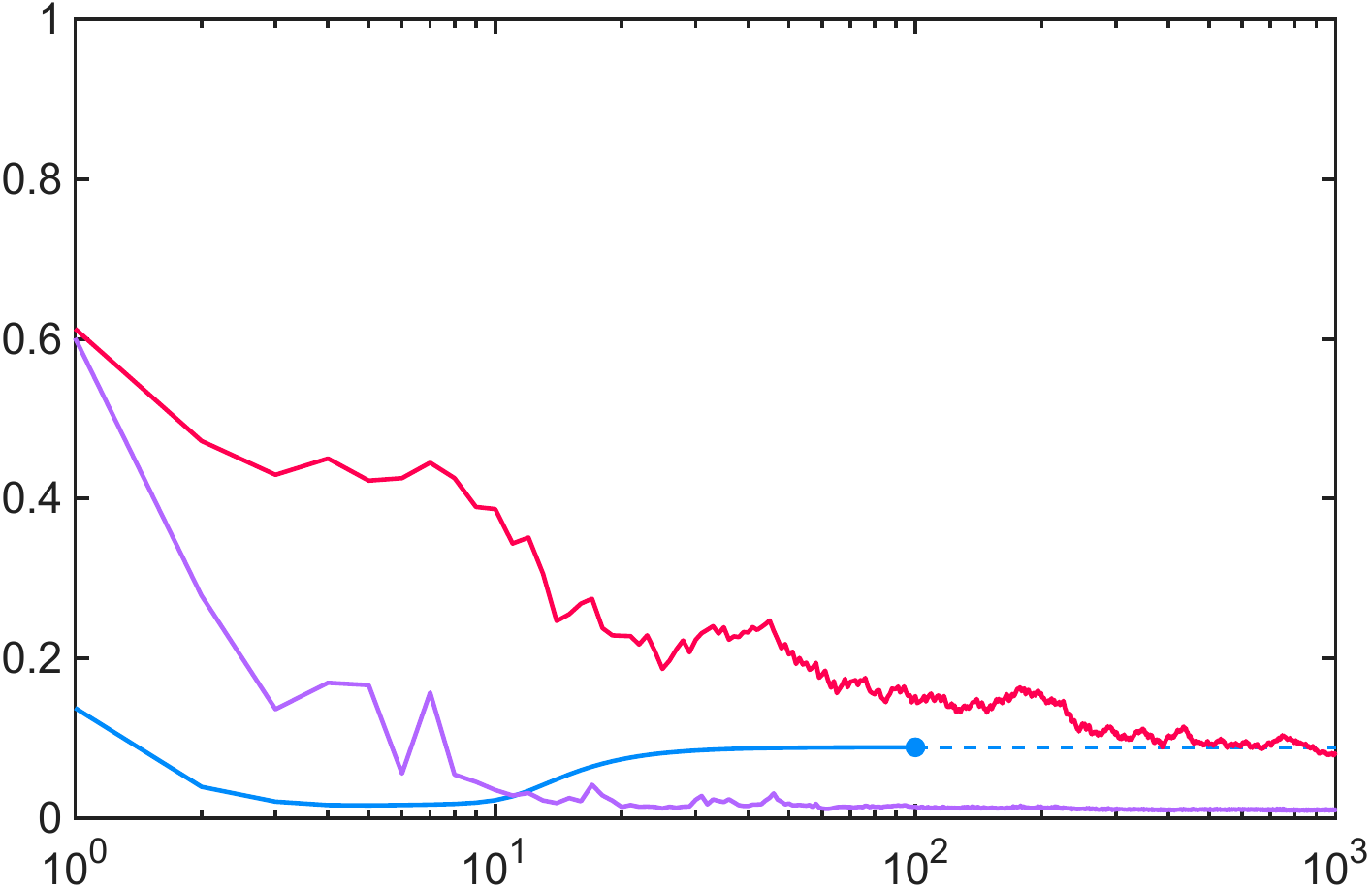}
    \end{minipage}\begin{minipage}{0.3\linewidth}
    \centering
        \footnotesize{30 \% i.i.d.}
        
        \includegraphics[width=0.98\linewidth]{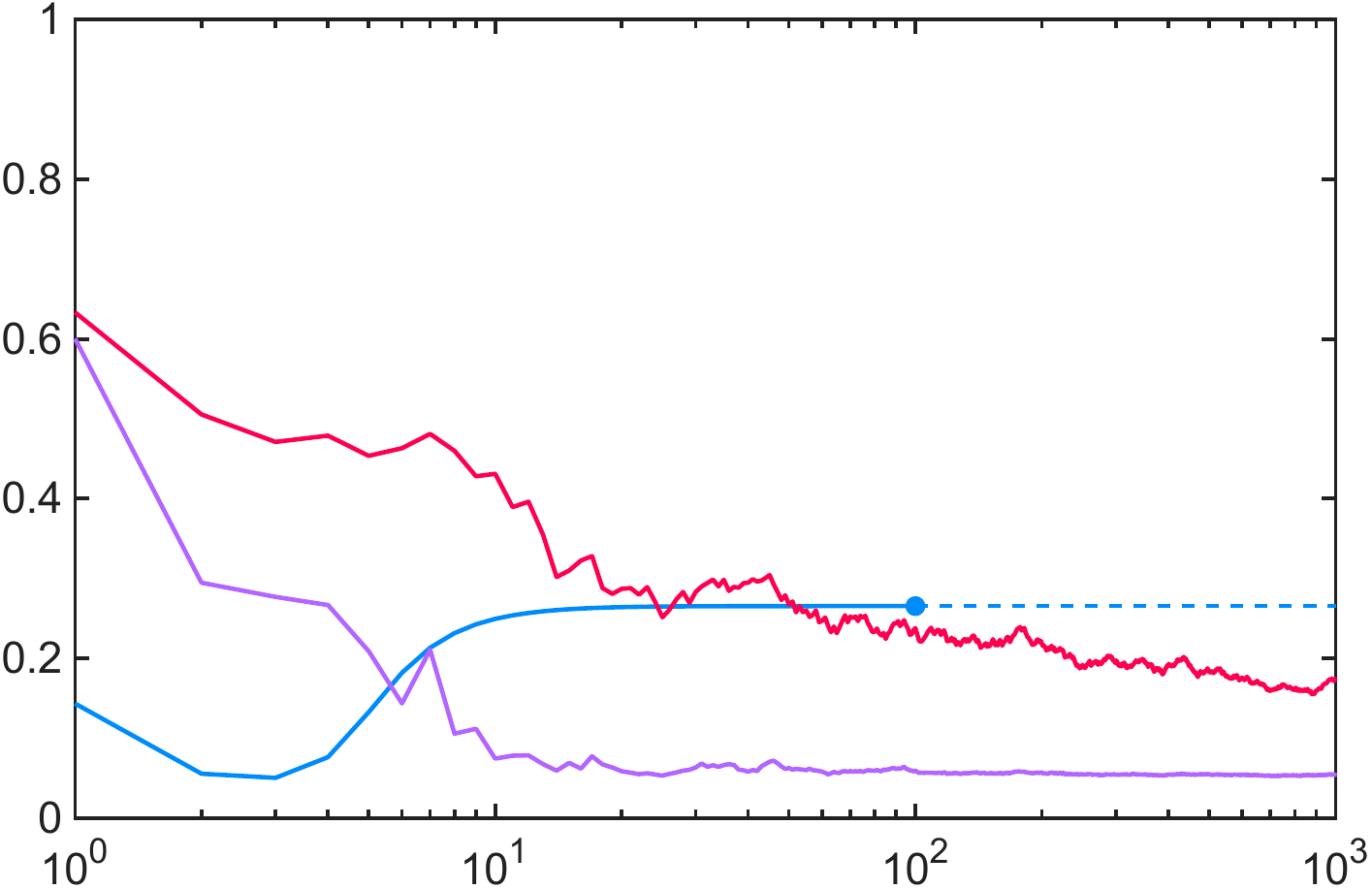}
    \end{minipage}\begin{minipage}{0.3\linewidth}
    \centering
        \footnotesize{30 \% spatially correlated}
        \includegraphics[width=0.98\linewidth]{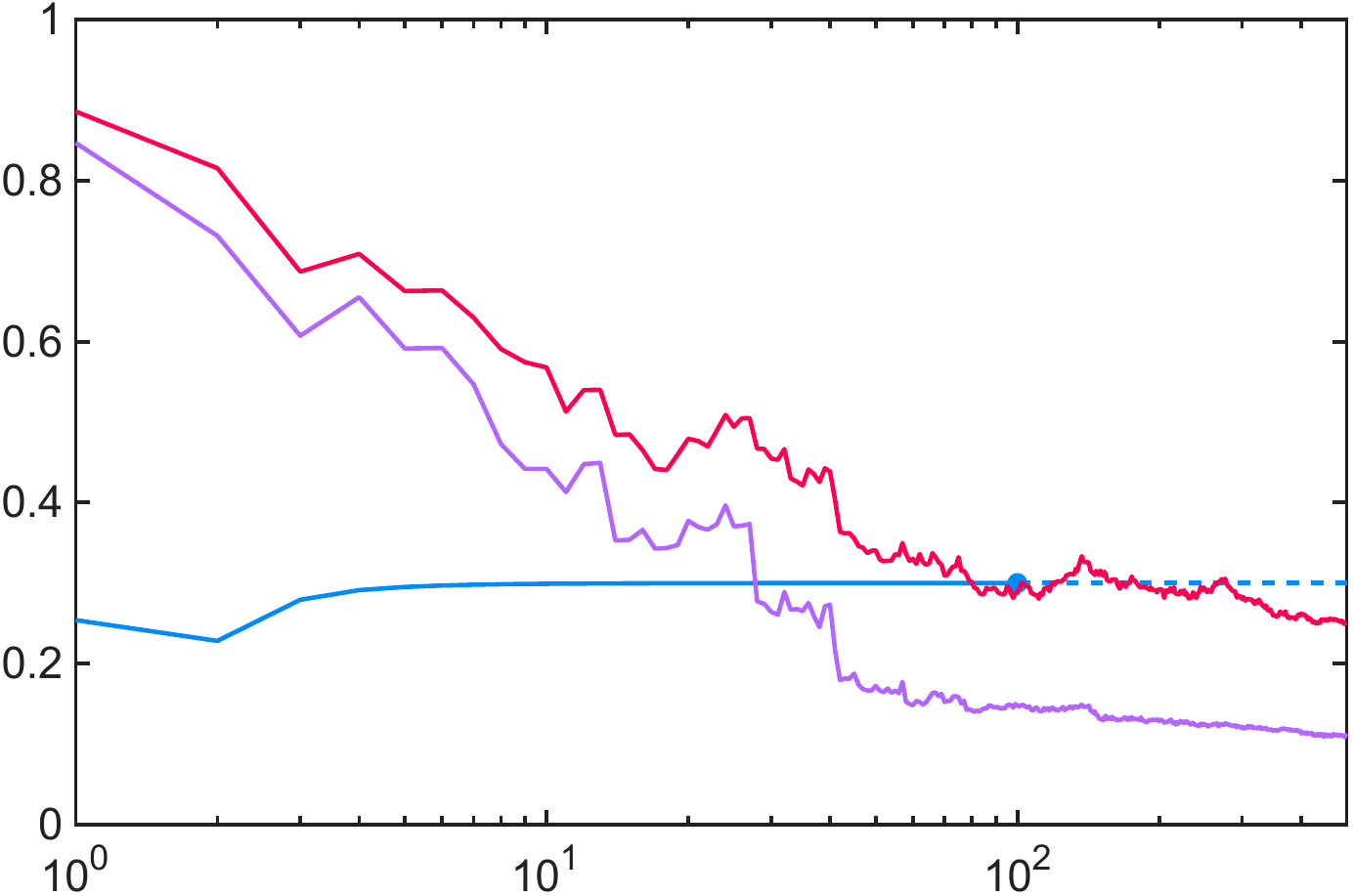}
    \end{minipage}
    \caption{Relative errors of the method for each foremost iteration. Iterations are presented on a logarithmic x-axis and errors on the y-axis. STVSB (light blue) is computed over 100 split Bregman iterations. UGE (red) and EC Bayesian (violet) are computed over 1,000 samples in 10 \% Gaussian i.i.d. noise case and otherwise limited to 500 samples due to convergence.}
    \label{fig:phantom_relativ_err_graphs2}
\end{figure}

\subsection{Experiment set 2}
The probability contour plotted in the homogeneous conductivity disk in Figure \ref{fig:SpatialProb} shows that, without information about the observation noise, the estimates are less likely to reconstruct far-field sources, even if the method is theoretically unbiased. This is because a deep source produces almost equal, weak signals in the sensor space. Thus, sensitivity to noise increases with depth. However, we can see that the source is almost guaranteed to be found from relatively deep, even with 15\% noise. From the figure, we can see that once the probability starts to decrease, it decreases rapidly around 0.2. This lower-hemisphere region is significantly smaller with 5 \% noise than with 15 \% noise. 

\begin{figure}[h!]
    \centering
    \begin{minipage}{0.3\textwidth}
    \begin{center}
        5 \% noise
    \end{center}\vskip0.02cm
        \includegraphics[width=0.95\textwidth]{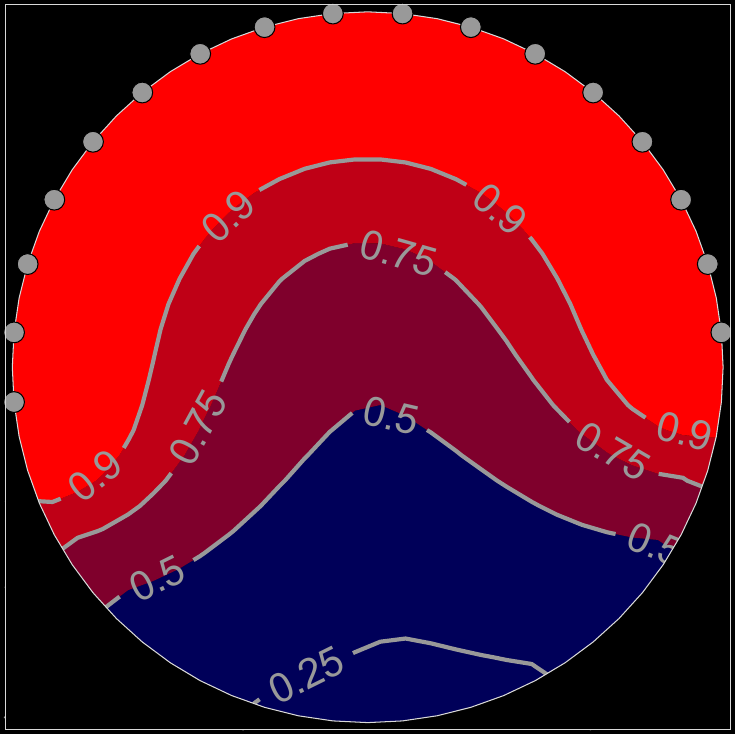}
    \end{minipage}\begin{minipage}{0.3\textwidth}
    \begin{center}
        15 \% noise
    \end{center}\vskip0.02cm
        \includegraphics[width=0.95\textwidth]{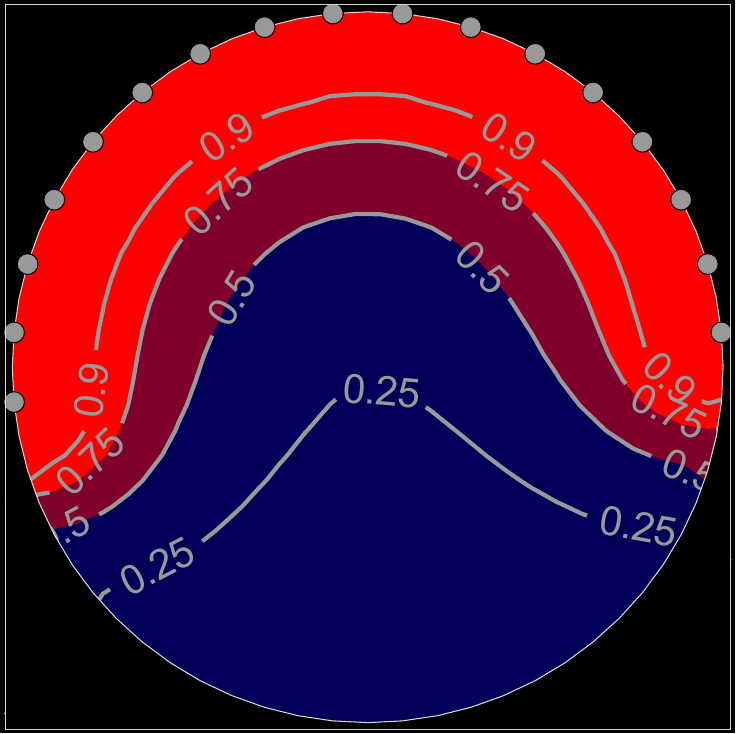}
    \end{minipage}
    \caption{Probability lower bound contour map to reconstruct exactly the source in the conductivity disk without knowledge of the observation noise, where the sensors lie only on one half of the disk (gray circles on the upper half). For the model, the degeneracy constant is $c=1$, and we computed mutual coherence for each candidate separately.}
    \label{fig:SpatialProb}
\end{figure}

By applying the notion of the signal-to-noise ratio (SNR) according to Eq. (\ref{eq:SNR}) to the probability bound, we can observe SNR-, sensor-, and source-parameter-amount-dependent behavior shown in Figure \ref{fig:FdistSNR}, where the probability of weakly reconstructing the source is plotted against SNR for four different sensor counts: 16, 32, 61, and 128 for one and three-parameter models using zero mutual coherence and such mutual coherence $\mu$ that $(m-d-1)\mu=0.5$ for $m$ sensors in dimension $d$. Unsurprisingly, comparing the two graphs shows that more parameters require a higher SNR for successful reconstruction. More interestingly, at lower SNR, fewer sensors give more reliable results: when mutual coherence is zero, with a one-parameter model, SNR below 2 (50 \% noise) favors 16 sensors over 128, and below 4 (25 \% noise) for the three-parameter model. Similarly, with the non-zero example mutual coherence, SNR below 7.2 (around 14 \% noise) is observed for the one-parameter model, and SNR below 20 (5 \% of noise) is observed with the three-parameter model. Thus, mutual coherence only increases the interval for which the lower sensor number is favored.
In addition, without much exaggeration, we can say that there are SNR limits beyond which one can or cannot reconstruct a source with high-density sensor systems. 

The estimated (A) one deep source, (B) highly superficial source pair causing orthogonal observations, and (C) near- and far-field source pair are presented in Figure \ref{fig:dipoles_in_disk}, where a red cross indicates the location and an arrow indicates the 2D orientation of the source. The results show that MNE and both sLORETA variants mislocalize the source, whereas UGE estimates the source at the nearest location allowed by the inversion model's mesh. The source orientation is also nearly correct. MNE localizes the source at the nearest observation boundary of the model with respect to the true source location due to the depth bias. sLORETA estimate is superficial, but not in the vicinity of the sensors. The mislocalization of sLORETA can be explained by the near-tangential orientation of the source relative to the nearest source. The source is radial for sources on the left part of the disk; however, the distance between the source and these sensors is high. sLORETA2D performs second-best, with a slight orientation error and a smaller localization error, as the estimate is farther from the boundary in the correct direction. By looking at the Z-score distribution in Figure \ref{fig:recs_in_disk}, the Z-scores of both sLORETAs are over 0.8. Estimation task (B) demonstrates that only sLORETA fails to distinguish between the two distant sources, as expected given the method's 1-parameter nature. Interestingly, the 2-orientation extension of the classical sLORETA, sLORETA2D, can obtain two sources. The Z-score distribution indicates another source on the right side of the model; however, that inference requires human input. sLORETA2D provides the best orientation estimate.
In scenario C, most methods indicate two nearby sources, with the more superficial one closer. In contrast, UGE indicates two far-field sources that are spatially relatively close to the true sources. By ignoring the significant localization error of the deep source with sLORETA2D, its orientation estimate is nearly exact. Based on the distributions, both sources have Z-scores above 0.8 for UGE, while sLORETA provides an almost 0.9 score for the cortex and over 0.6 for the deep source. From the sLORETA2D distribution, it may not be apparent that a source lies deeper than the disk's midpoint.

\begin{figure}
    \centering
    \rotatebox{90}{\hspace{0.75cm} 3 parameter}\rotatebox{90}{\hspace{1.25cm} $\mu=0$}\hspace{0.1cm}\includegraphics[width=0.5\linewidth]{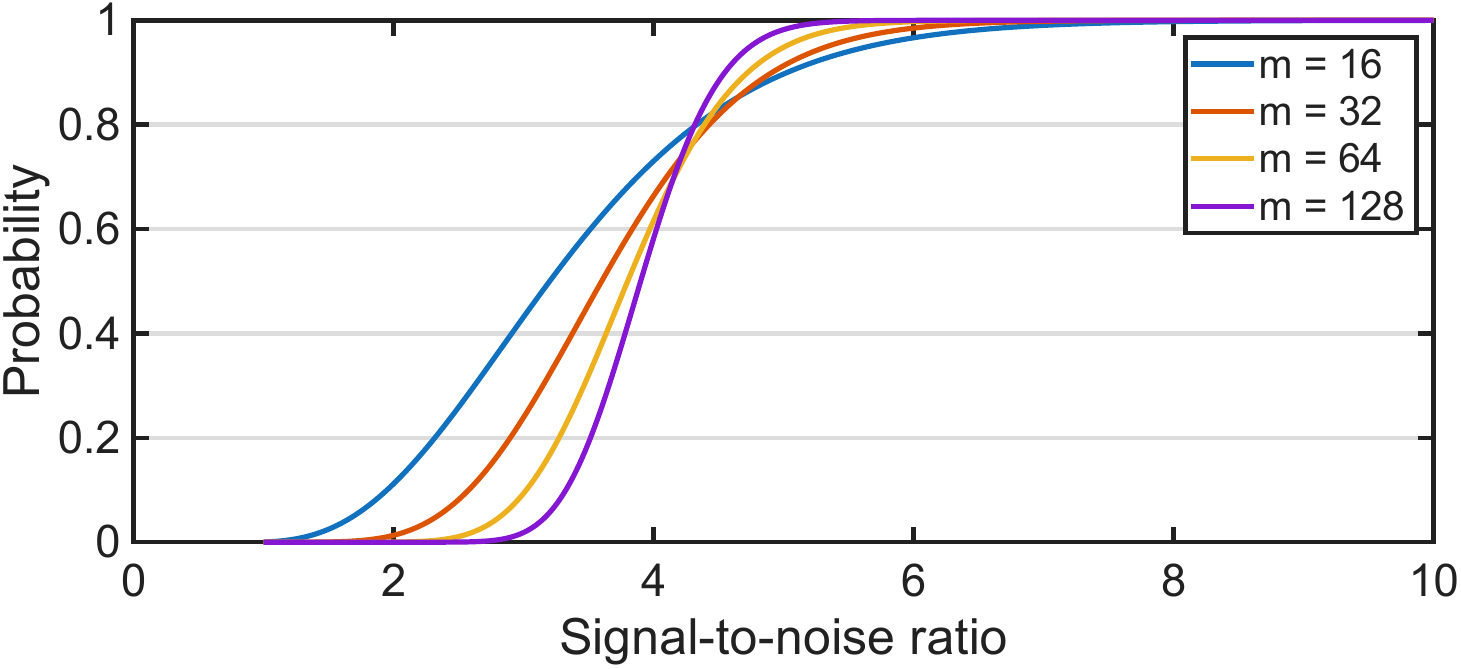}

    \rotatebox{90}{\hspace{0.75cm} 1 parameter}\rotatebox{90}{\hspace{1.25cm} $\mu=0$}\hspace{0.1cm}\includegraphics[width=0.5\linewidth]{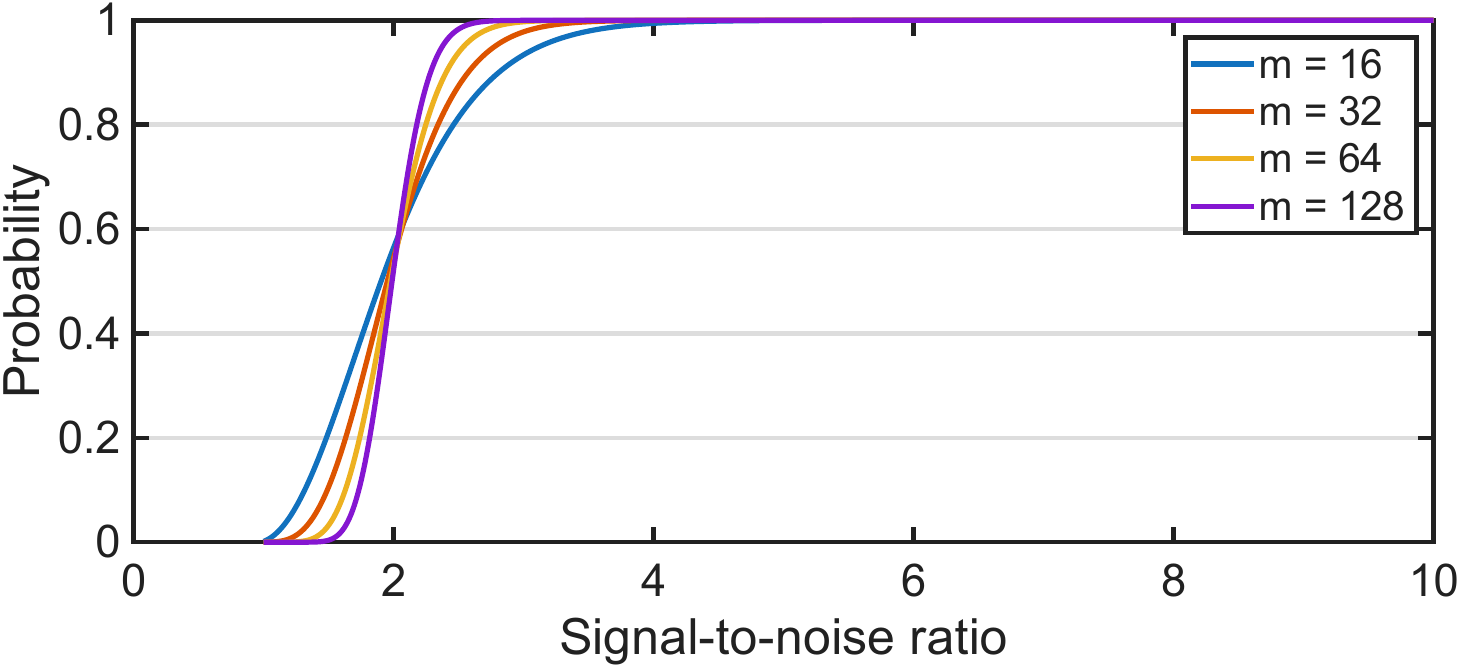}

    \rotatebox{90}{\hspace{0.75cm} 3 parameter}\rotatebox{90}{\hspace{0.5cm} $(m-4)\mu=0.5$}\hspace{0.1cm}\includegraphics[width=0.5\linewidth]{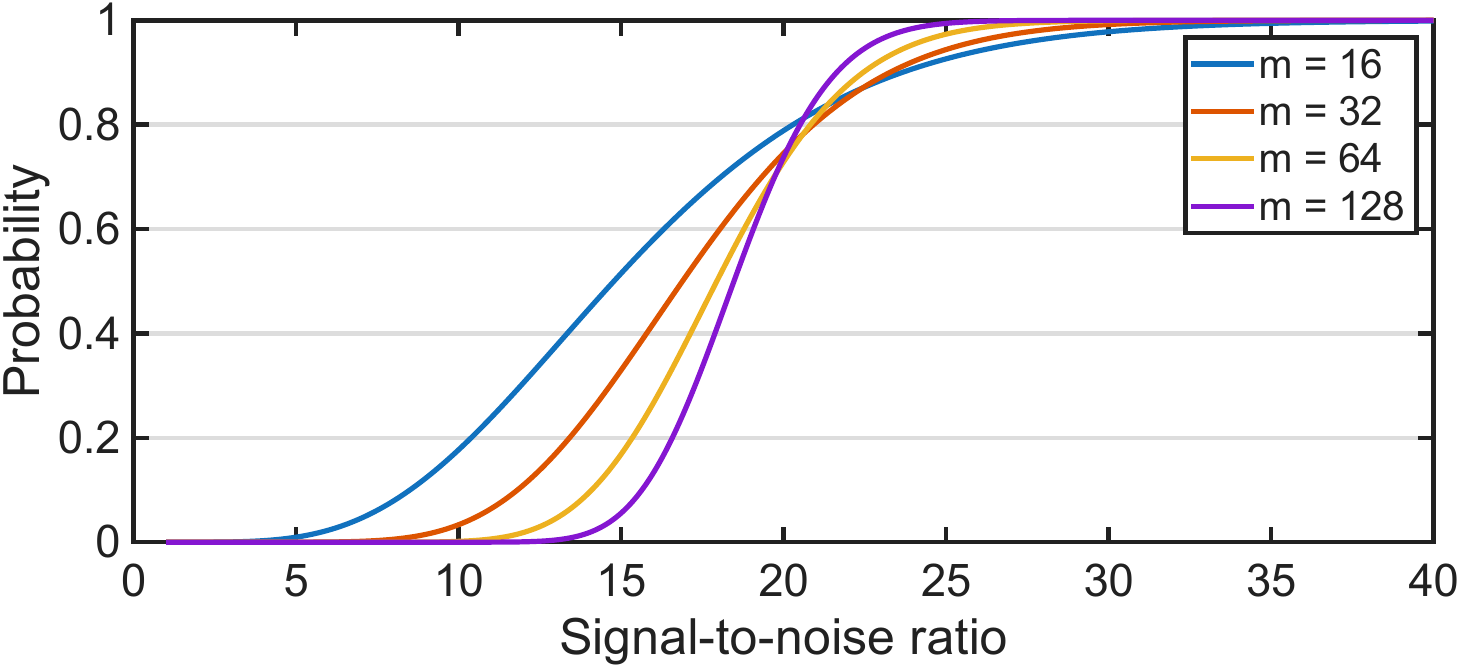}

    \rotatebox{90}{\hspace{0.75cm} 1 parameter}\rotatebox{90}{\hspace{0.5cm} $(m-2)\mu=0.5$}\hspace{0.1cm}\includegraphics[width=0.5\linewidth]{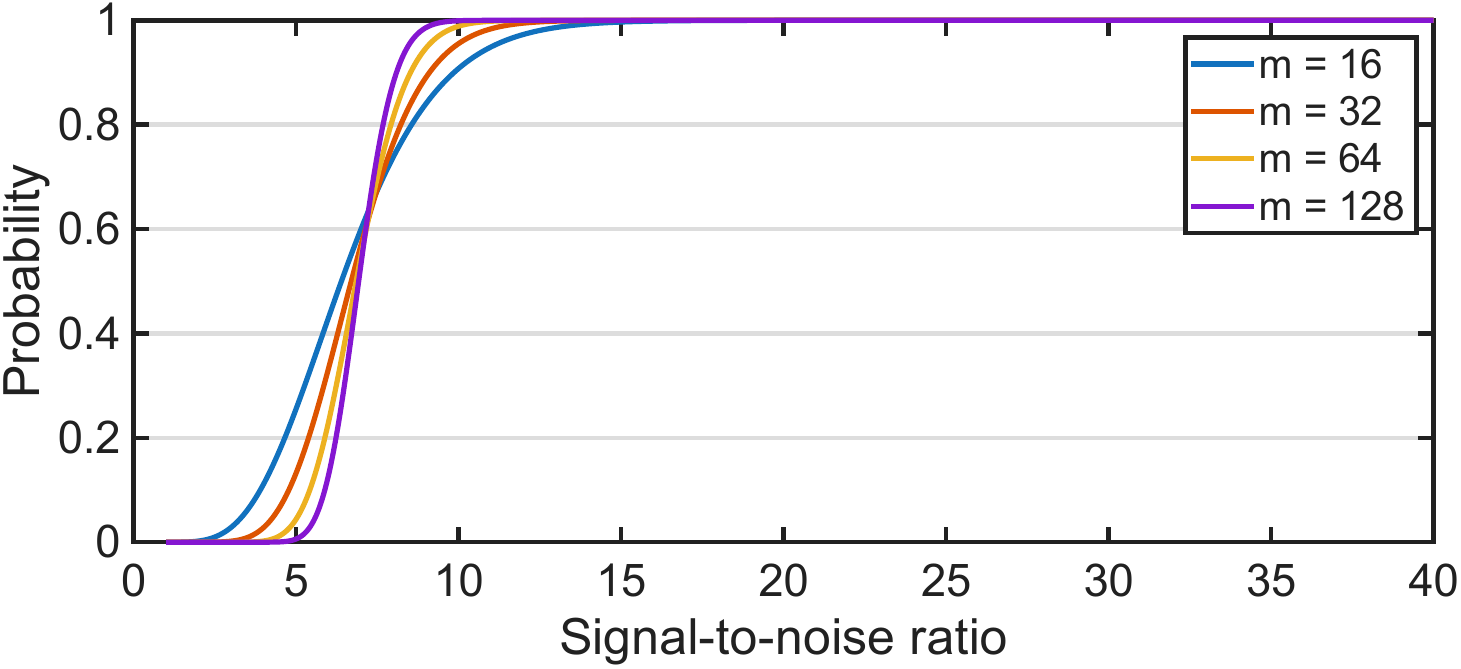}
    \caption{Reconstruction probability lower bound under different levels of signal-to-noise ratio as defined in Eq. (\ref{eq:SNR}), assuming zero mutual coherence for the first two graphs starting from the top, and $(m-d-1)\mu=0.5$ for the latter two graphs. The legend shows probability curves for four sensor counts: 16, 32, 64, and 128. {\bf First and third} graphs show the probabilities for a model where the source is described by three parameters, e.g., sources have free orientation in a three-dimensional space, and {\bf second and fourth} graphs show the same for a 1-parameter-per-source model.}
    \label{fig:FdistSNR}
\end{figure}

\begin{figure}
    \centering
    \includegraphics[width=0.9\linewidth]{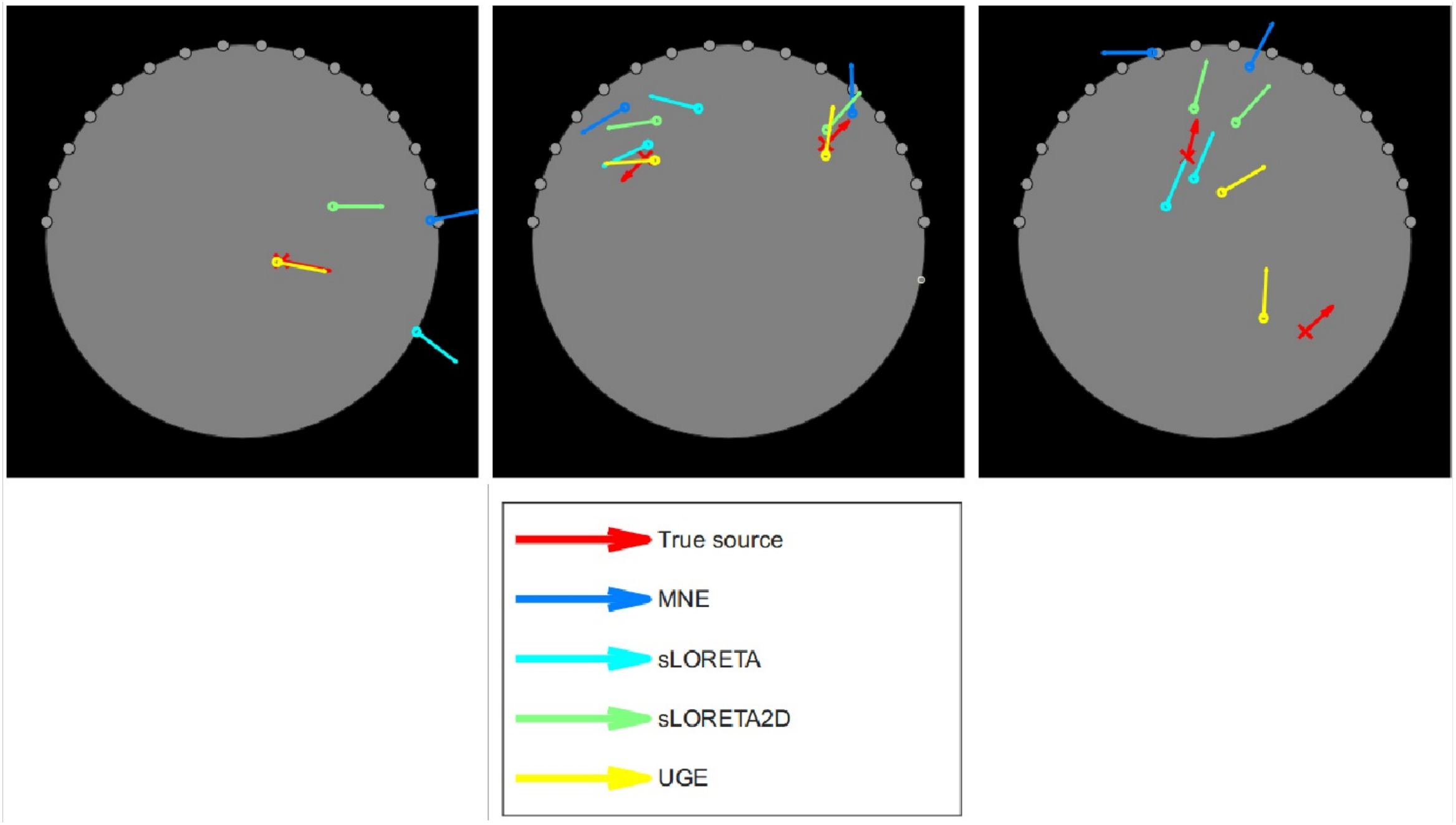}
    \caption{Dipolar source estimations of MNE, sLORETA, sLORETA2D, and the Unbiased Bayesian method (UGE) in the homogeneous conductivity disk in three different source scenarios A, B, and C. The red cross indicates the location of the true source, and the red arrow represents the 2-dimensional source orientation. Locations of the source estimates are represented by colored circles and orientations by the arrows of the same color. The legend shows the connections between colors and methods. }
    \label{fig:dipoles_in_disk}
\end{figure}

\begin{figure}
    \centering
    \includegraphics[width=0.95\linewidth]{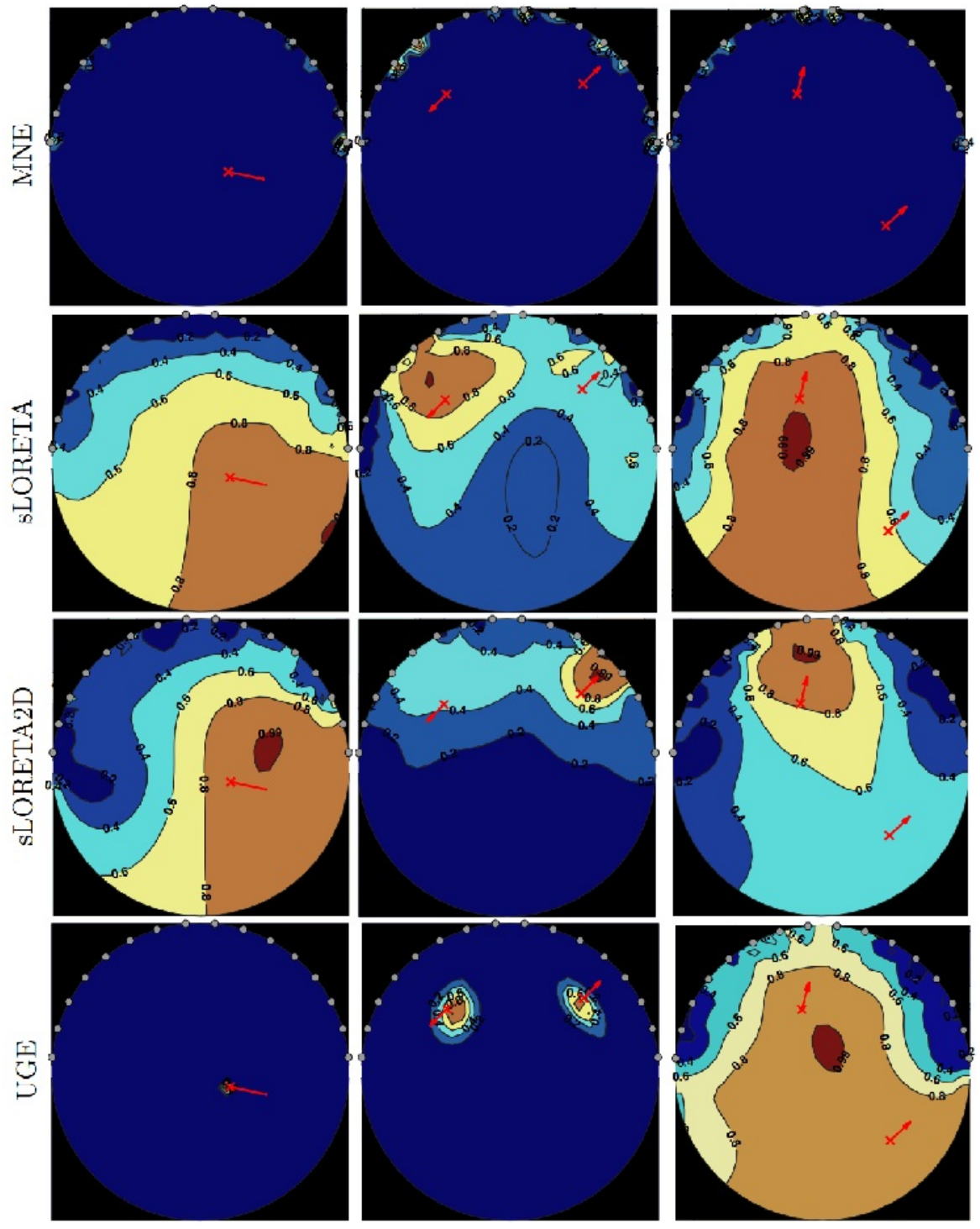}
    \caption{The distributed source estimate of MNE and Z-score distribution of standardized methods and unbiased Gaussian estimate (UGE) displayed as contours on the conductivity disk for three different source scenarios. The true source is depicted by a red cross and arrow.}
    \label{fig:recs_in_disk}
\end{figure}

\clearpage

\subsection{Experiment set 3}
In this experiment, we aim to estimate two simultaneous sources corresponding to Experiment 2.B, in which one source is at a superficial location and the other lies within deep brain structures. In Figure \ref{fig:recsonbrain}, the Z-score distributions are plotted over the transversal, coronal, and sagittal plane cuts of MRI images. The images are split into two rows for each method. The upper row presents the score distribution at the planes where the superficial, near-field source is located, and the second row displays the score distribution at the planes of the deep, far-field source.

As the Z-score distributions show, neither sLORETA nor sLORETA3D forms separate clusters of highest scores around the true sources; instead, the highest-value distribution (depicted in yellow) forms a continuous region that covers both sources. Because these two score distributions are highly similar and cover a large region around both sources, it is difficult to estimate the exact source location. However, sLORETA3D appears to yield a slightly higher score in the region adjacent to the deep source at the posterior end of the thalamus.

In contrast, UGE produces separate blobs at the superficial source and near the deep source. The maximum score in the deep structures is lower at the posterior end of the thalamus, as in both sLORETA variants.

\begin{figure}
    \centering
    \includegraphics[width=0.85\linewidth]{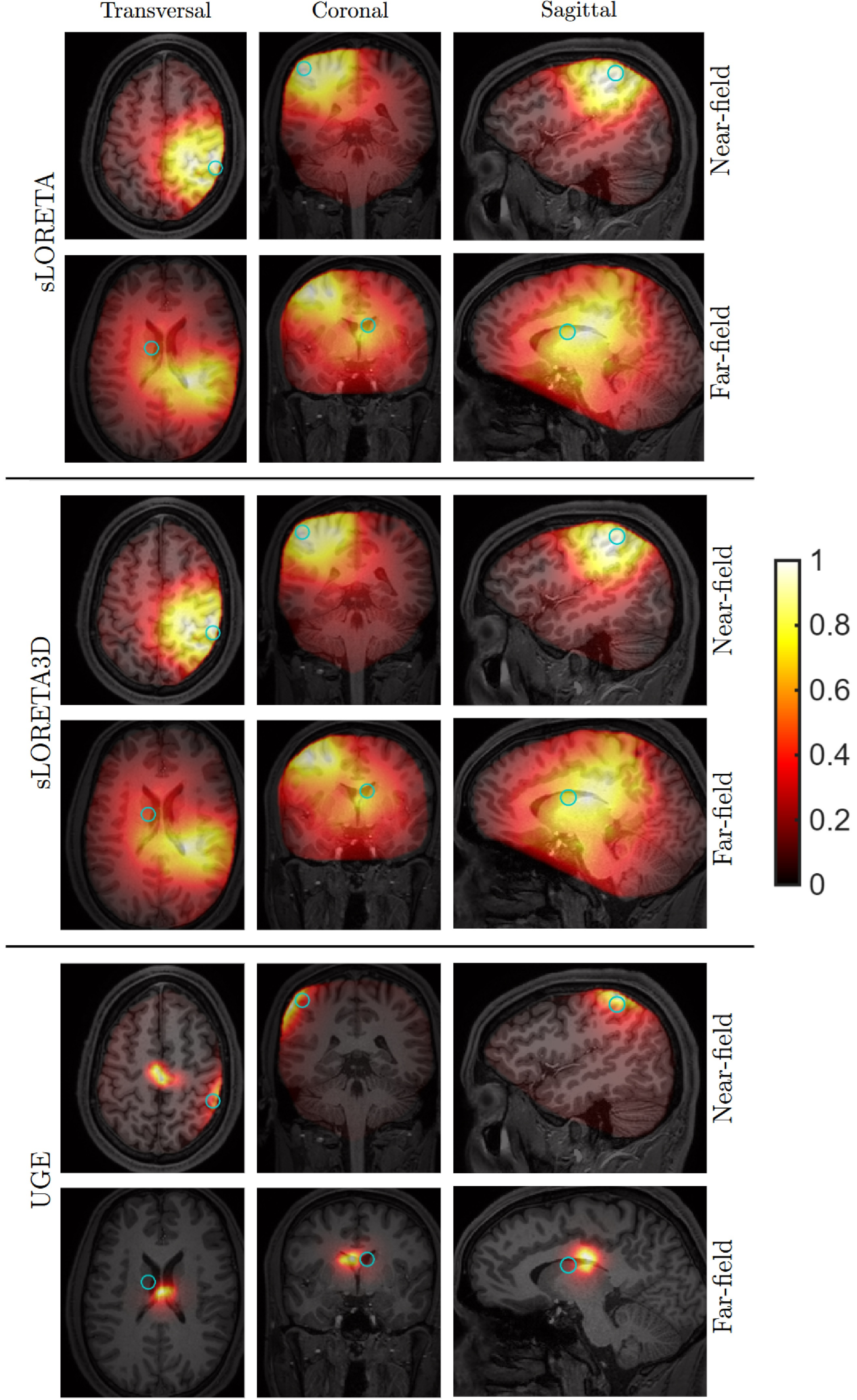}
    
    \caption{Z-score distributions of sLORETA, sLORETA3D, and UGE interpolated over the MRI slices, where the superficial near-field source and deep far-field source are located in subsequent image rows. The turquoise circle indicates the location of the corresponding source.}
    \label{fig:recsonbrain}
\end{figure}

\clearpage

\section{Discussion}
In this paper, we generalize the theory of unbiased localization achieved by standardization \cite{PascualMarqui2002,Lahtinen2024} to unbiased parameter estimation and signal reconstruction, i.e., generalizing unbiased outputs from one to an arbitrary number. Furthermore, we draw a strong connection between the proposed unbiased estimation framework and compressed sensing. Concepts from compressed sensing allow us to analyze the estimation technique in a completely new light. The generalization provides an answer to the common criticism of standardized methods, namely that the Z-scores are not the actual solution to the inverse problem: The new perspective shows us that if we think of the recoverable object as a sparse vector field, the unbiased estimate provides the regions where the field is non-zero, while a fitting strategy must be used to recover the actual solution. This is congruent with the classical dipole fitting approach, where a vector from the field is recovered by the fitting procedure, and the most likely vectors are decided by the goodness-of-fit measure. In fact, it has been proven that there exists an isomorphism between sLORETA and any dipole fitting approach, and these estimates are the same in one-source scenarios \cite{HoltershinkenOnEquivalence}. Moreover, a least-squares fitting or a Gaussian fitting approach has no depth bias if the covariance structure of the corresponding Bayesian model's innovation covariance matches the observation covariance, i.e., $\mathrm{cov}\left[{\bf y}\right]\propto \mathrm{cov}_{({\bf x},{\bf n})}\left[L{\bf x}+{\bf n}\right]$.

Next, we presented theorems outlining how many parameters can be recovered exactly in noiseless cases, based on the properties of the system matrix or lead field. Exact reconstruction under observation noise was not possible; hence, we divided the conditions for exact reconstruction into strongly and weakly reconstructible parameters, with the latter relaxing the requirements away from exact reconstruction. In addition, we found the probability measure for weakly reconstructible parameters, meaning the true contributors to the observed signals are more pronounced than the nuisance caused by observation noise. By coupling the measure with the concept of signal-to-noise ratio, we conclude that standardized methods are robust to significant measurement noise because the signal-to-noise ratio makes the noncentral $F$-distributed variable independent of the depth of the true source. However, our results show a trade-off between the number of sensors and the signal-to-noise ratio. This could give the mathematical explanation for the conflicting results with practical brain source imaging and the recommended number of electroencephalography electrodes \cite{Sohrabpour2015Electrodenum,Montoya2021Electrodenum,Moritz2022Electrodenum} as well as the motivation for the equal or improved source estimation with well-targeted sensors over high-density sensors \cite{SolerAndres2022ElectrodenumTargeted,Horrillo-Maysonnial2023ElectrodenumTargeted} when the measurements are highly noisy.

The Gaussian model assumption and fixed discrete source positions yield a new unbiased method we call the Unbiased Gaussian Estimate (UGE), a generalization of sLORETA. We tested the method against the Spatial Total Variation Split Bregman (STVSB) algorithm \cite{Montesinos2014STVSB} for noisy phantom image reconstruction, and against sLORETA and its parameter variants for source estimation in a two-dimensional conductivity disk and a 3-dimensional multicompartment head model. The results demonstrate the properties of UGE predicted by the given theorems. While UGE is highly impractical for image recovery, unlike STVSB, it can produce a distinguishable phantom even when the problem is not sparse, and it does not reconstruct as much noise into the image as STVSB. Moreover, the variation with a priori knowledge of the ratios between colors is able to recover the phantom from highly noisy data, almost exactly, once the noise is identically and independently distributed, and the number of line samples agrees with the requirements for unique recovery in the noiseless case. Interestingly, the number of samples provided to STVSB has a vanishingly small impact on recovery quality, which may be a byproduct of its emphasis on maintaining sharp edges between neighboring contrast pixels. 

The source imaging results highlight the challenges that standardized methodologies face in estimating multiple sources \cite{Lahtinen2024}. Although the high-Z-score region often covers both source locations, deducing their exact locations is difficult or even impossible. Because UGE's multiparameter framework supports multisource reconstruction, it provides much clearer source estimates and is more accurate than both sLORETA estimates. The superiority of UGE is not surprising, as its Bayesian model better matches the problem of recovering sources with multidimensional orientation.

The paradigm shift from post-hoc weighting to modifying the system matrix also affects the development and software implementation of unbiased methods: the previous scheme focused on transforming a regular estimate into a standardized estimate via post-hoc weights. Defining the transformation can be complicated when the Bayesian formulation is unknown or not Gaussian; for example, SSLOFO \cite{HeshengLiu2005SsL} uses fixed-point iterations and a shrinkage operation to determine the post-hoc weights for a non-Gaussian model. Spatiotemporal standardization with Kalman filtering is a major computational bottleneck \cite{LahtinenJoonas2026StableSKF,Dilshanie2026SKFParameters}. Singular value decomposition can bypass these theoretical and computational difficulties by making the system matrix unbiased. Furthermore, recent developments in Bayesian source imaging already make use of the right-hand-side singular sub-matrices to improve grouped parameter estimation for non-Gaussian priors \cite{Lahtinen2025SVD}, such as Laplace-distributed dipolar sources with 2- or 3-dimensional orientations. Modifying the system matrix with a large spectrum for sparsity-promoting inverse algorithms could be numerically more accurate than using an iterative algorithm to incrementally reduce bias and correct the transformation. 

As future work, we are looking for new applications where unbiased estimation could be useful. In addition, we are developing a new method that utilizes singular value decomposition both to find an unbiased solution and to refine the parameter estimation. Technically speaking, extending unbiased estimation to continuous problems looks more promising than ever.

\section*{Data availability statement}
The data that support the findings of this study are available upon reasonable request from the
author. The Dataset of the Ernie subject is available at the SimNIBS 4 GitHub page: \url{https://simnibs.github.io/simnibs/build/html/index.html}.

\section*{Funding}
The work of J.\ Lahtinen was supported by ERA PerMED project "Personalised diagnosis and treatment for refractory focal paediatric and adult epilepsy (PerEpi)" (AoF number 344712). The author would also like to thank the German Academic Exchange Service (DAAD) for their support (decision 354976).


\section*{References}

\bibliography{sample}

\begin{thebibliography}{10}

\bibitem{MNE_cardio_biased}
T.~A. Crowley, C.~D. Haupt, and D.~B. Kynor.
\newblock A {W}eighting {M}atrix to {R}emove {D}epth {B}ias in the {L}inear {B}iomagnetic {I}nverse {P}roblem with {A}pplication to {C}ardiology.
\newblock In Cheryl~J. Aine, Gerhard Stroink, Charles~C. Wood, Yoshio Okada, and Stephen~J. Swithenby, editors, {\em Biomag 96}, pages 197--200, New York, NY, 2000. Springer New York.

\bibitem{PascualMarqui2002}
R~D Pascual-Marqui.
\newblock Standardized low-resolution brain electromagnetic tomography (s{LORETA}): technical details.
\newblock {\em Methods and findings in experimental and clinical pharmacology}, 24:5--12, 2002.

\bibitem{Manfred1999}
Manfred Fuchs, Michael Wagner, Thomas Köhler, and Hans-Aloys Wischmann.
\newblock Linear and {N}onlinear {C}urrent {D}ensity {R}econstructions.
\newblock {\em Journal of clinical neurophysiology}, 16(3):267--295, 1999.

\bibitem{SekiharaNagarajan2008}
Kensuke. Sekihara and Srikatan~S. Nagarajan.
\newblock {\em Adaptive Spatial Filters for Electromagnetic Brain Imaging}.
\newblock Series in Biomedical Engineering. Springer Berlin Heidelberg, Berlin, Heidelberg, 1st ed. 2008. edition, 2008.

\bibitem{SamuelssonJohnMaximumPrinciple4DistrMethds2020}
John~G. Samuelsson, Noam Peled, Fahimeh Mamashli, Jyrki Ahveninen, and Matti~S. Hämäläinen.
\newblock Spatial fidelity of {MEG}/{EEG} source estimates: A general evaluation approach.
\newblock {\em NeuroImage (Orlando, Fla.)}, 224:117430--117430, 2021.

\bibitem{Grech2008}
Roberta Grech, Tracey Cassar, Joseph Muscat, Kenneth~P Camilleri, Simon~G Fabri, Michalis Zervakis, Petros Xanthopoulos, Vangelis Sakkalis, and Bart Vanrumste.
\newblock Review on solving the inverse problem in {EEG} source analysis.
\newblock {\em Journal of neuroengineering and rehabilitation}, 5(1):25--25, 2008.

\bibitem{Calvetti2009}
Daniela Calvetti, Harri Hakula, Sampsa Pursiainen, and Erkki Somersalo.
\newblock Conditionally {G}aussian {H}ypermodels for {C}erebral {S}ource {L}ocalization.
\newblock {\em SIAM Journal on Imaging Sciences}, 2(3):879--909, 2009.

\bibitem{GravedePeralta009NEIP}
Rolando Grave~de Peralta, Olaf Hauk, Sara~L. Gonzalez, and Fabio Babiloni.
\newblock The neuroelectromagnetic inverse problem and the zero dipole localization error.
\newblock {\em Computational Intelligence and Neuroscience}, 2009(1):137--147, 2009.

\bibitem{GuofaShou2007StEI}
Guofa Shou, Ling Xia, and Mingfeng Jiang.
\newblock Solving the {E}lectrocardiography {I}nverse {P}roblem by {U}sing an {O}ptimal {A}lgorithm {B}ased on the {T}otal {L}east {S}quares {T}heory.
\newblock In {\em Third International Conference on Natural Computation (ICNC 2007)}, volume~5, pages 115--119. IEEE, 2007.

\bibitem{CarreraJess2005hydrogeology}
Jes~s Carrera, Andr~s Alcolea, Agust~n Medina, Juan Hidalgo, and Luit~J. Slooten.
\newblock Inverse problem in hydrogeology.
\newblock {\em Hydrogeology journal}, 13(1):206--222, 2005.

\bibitem{JohnsonGeoPhys2015}
T.C Johnson and D~Wellman.
\newblock Accurate modelling and inversion of electrical resistivity data in the presence of metallic infrastructure with known location and dimension.
\newblock {\em Geophysical journal international}, 202(2):1096--1108, 2015.

\bibitem{ShiWenyangGeoPhys2023}
Wenyang Shi, Guangzhi Yin, Mi~Wang, Lei Tao, Mengjun Wu, Zhihao Yang, Jiajia Bai, Zhengxiao Xu, and Qingjie Zhu.
\newblock Progress of electrical resistance tomography application in oil and gas reservoirs for development dynamic monitoring.
\newblock {\em Processes}, 11(10):2950--, 2023.

\bibitem{Lahtinen2024}
Joonas Lahtinen.
\newblock On bias and its reduction via standardization in discretized electromagnetic source localization problems.
\newblock {\em Inverse Problems}, 40(9):095002, July 2024.

\bibitem{Lahtinen2024SKF}
Joonas Lahtinen, Paavo Ronni, Narayan~Puthanmadam Subramaniyam, Alexandra Koulouri, Carsten Wolters, and Sampsa Pursiainen.
\newblock Standardized {K}alman filtering for time serial source localization of simultaneous subcortical and cortical brain activity.
\newblock \url{www.techrxiv.org/722443/UqPLPdTHCJ-oxYE2bCoI6Q}, 2024.

\bibitem{WagnerMichael2004sLORscr}
Michael Wagner, Manfred Fuchs, and Jörn Kastner.
\newblock Evaluation of s{LORETA} in the presence of noise and multiple sources.
\newblock {\em Brain topography}, 16(4):277--280, 2004.

\bibitem{LealAlberto2008sLORscr}
Alberto~J.R Leal, Ana~I Dias, José~P Vieira, Ana Moreira, Luís Távora, and Eulália Calado.
\newblock Analysis of the dynamics and origin of epileptic activity in patients with tuberous sclerosis evaluated for surgery of epilepsy.
\newblock {\em Clinical neurophysiology}, 119(4):853--861, 2008.

\bibitem{LiChunsheng2016_30threshold}
Chunsheng Li, Daniel Jacobs, Trevor Hilton, Martin~del Campo, Yotin Chinvarun, Peter~L Carlen, and Berj~L Bardakjian.
\newblock Epileptogenic source imaging using cross-frequency coupled signals from scalp eeg.
\newblock {\em {IEEE} Trans. Biomed. Eng.}, 63(12):2607--2618, 2016.

\bibitem{Dumpelmann2012sLORepil}
Matthias Dümpelmann, Tonio Ball, and Andreas Schulze-Bonhage.
\newblock s{LORETA} allows reliable distributed source reconstruction based on subdural strip and grid recordings.
\newblock {\em Human brain mapping}, 33(5):1172--1188, 2012.

\bibitem{deGooijer2013_14}
Karin~L. de~Gooijer-van~de Groep, Frans~S.S. Leijten, Cyrille~H. Ferrier, and Geertjan~J.M. Huiskamp.
\newblock Inverse modeling in magnetic source imaging: Comparison of {MUSIC}, {SAM}(g2), and s{LORETA} to interictal intracranial {EEG}.
\newblock {\em Human brain mapping}, 34(9):2032--2044, 2013.

\bibitem{Coito2019_18}
Ana Coito, Silke Biethahn, Janina Tepperberg, Margherita Carboni, Ulrich Roelcke, Margitta Seeck, Pieter Mierlo, Markus Gschwind, and Serge Vulliemoz.
\newblock Interictal epileptogenic zone localization in patients with focal epilepsy using electric source imaging and directed functional connectivity from low‐density {EEG}.
\newblock {\em Epilepsia open}, 4(2):281--292, 2019.

\bibitem{LiRui2021_7}
Rui Li, Chris Plummer, Simon~J. Vogrin, William~P. Woods, Levin Kuhlmann, Ray Boston, David~T.J. Liley, Mark~J. Cook, and David~B. Grayden.
\newblock Interictal spike localization for epilepsy surgery using magnetoencephalography beamforming.
\newblock {\em Clinical neurophysiology}, 132(4):928--937, 2021.

\bibitem{KimKwangYeon2022sLORepil}
Kwang~Yeon Kim, Ja-Un Moon, Joo-Young Lee, Tae-Hoon Eom, Young-Hoon Kim, and In-Goo Lee.
\newblock Distributed source localization of epileptiform discharges in juvenile myoclonic epilepsy: Standardized low-resolution brain electromagnetic tomography (sloreta) study.
\newblock {\em Medicine (Baltimore)}, 101(26):e29625--e29625, 2022.

\bibitem{vandeVelden2023}
Daniel van~de Velden, Ev-Christin Heide, Caroline Bouter, Jan Bucerius, Christian~H. Riedel, and Niels~K. Focke.
\newblock Effects of inverse methods and spike phases on interictal high-density {EEG} source reconstruction.
\newblock {\em Clinical neurophysiology}, 156:4--13, 2023.

\bibitem{Pascual2007sqrtm}
Roberto~D Pascual-Marqui.
\newblock Discrete, {3D} distributed, linear imaging methods of electric neuronal activity. {P}art 1: exact, zero error localization.
\newblock {\em arXiv.org}, 2007.

\bibitem{HoltershinkenOnEquivalence}
Malte~Bernhard Höltershinken, Tim-Rene Erdbrügger, Johannes Vorwerk, Christian Engwer, and Carsten~Hermann Wolters.
\newblock On the equivalence of methods for the eeg and meg inverse problems.
\newblock {\em Inverse Problems}, 2026.

\bibitem{BjorckA1983ASqrtm}
A~Björck.
\newblock A {S}chur method for the square root of a matrix.
\newblock {\em Linear algebra and its applications}, 52-53(1):127--140, 1983.

\bibitem{WelchL1974MutualCoherenceBound}
L.~Welch.
\newblock Lower bounds on the maximum cross correlation of signals (corresp.).
\newblock {\em IEEE transactions on information theory}, 20(3):397--399, 1974.

\bibitem{Donoho2003MutualCoherenceUniqBound}
David~L Donoho and Michael Elad.
\newblock Optimally sparse representation in general (nonorthogonal) dictionaries via l minimization.
\newblock {\em Proceedings of the National Academy of Sciences - PNAS}, 100(5):2197--, 2003.

\bibitem{ZhaoYun-Bin2013NewSparsestBounds}
Yun-Bin Zhao.
\newblock New and improved conditions for uniqueness of sparsest solutions of underdetermined linear systems.
\newblock {\em Applied mathematics and computation}, 224:58--73, 2013.

\bibitem{Calvetti2019AutommaticDepthWeighting}
Daniela Calvetti, Annalisa Pascarella, Francesca Pitolli, Erkki Somersalo, and Barbara Vantaggi.
\newblock Brain activity mapping from {MEG} data via a hierarchical {B}ayesian algorithm with automatic depth weighting.
\newblock {\em Brain topography}, 32(3):363--393, 2019.

\bibitem{Gershgorin1931}
S.~Ger{\v{s}}gorin.
\newblock {\"U}ber die abgrenzung der eigenwerte einer matrix.
\newblock {\em Izvestiya Akademii Nauk SSSR, Seriya Fiziko-Matematicheskikh Nauk}, 6:749--754, 1931.

\bibitem{CandesEmmanuel2006}
E.J. Candes, J.~Romberg, and T.~Tao.
\newblock Robust uncertainty principles: exact signal reconstruction from highly incomplete frequency information.
\newblock {\em IEEE transactions on information theory}, 52(2):489--509, 2006.

\bibitem{ChartrandR2008randomvstomographicsampling}
R.~Chartrand.
\newblock Nonconvex compressive sensing and reconstruction of gradient-sparse images: Random vs. tomographic fourier sampling.
\newblock In {\em 2008 15th IEEE International Conference on Image Processing}, pages 2624--2627. IEEE, 2008.

\bibitem{Goldstein2009SplitBregman}
Tom Goldstein and Stanley Osher.
\newblock The split bregman method for l1-regularized problems.
\newblock {\em SIAM journal on imaging sciences}, 2(2):323--343, 2009.

\bibitem{Montesinos2014STVSB}
Paula Montesinos, Juan Felipe~P.J. Abascal, Lorena Cussó, Juan~José Vaquero, and Manuel Desco.
\newblock Application of the compressed sensing technique to self-gated cardiac cine sequences in small animals.
\newblock {\em Magnetic resonance in medicine}, 72(2):369--380, 2014.

\bibitem{HamalainenMNE}
M.S Hämäläinen and R.J Ilmoniemi.
\newblock Interpreting magnetic fields on the brain: minimum norm estimates.
\newblock {\em Medical \& biological engineering \& computing}, 32(1):35--42, 1994.

\bibitem{PUONTI2020117044}
Oula Puonti, Koen {Van Leemput}, Guilherme~B. Saturnino, Hartwig~R. Siebner, Kristoffer~H. Madsen, and Axel Thielscher.
\newblock Accurate and robust whole-head segmentation from magnetic resonance images for individualized head modeling.
\newblock {\em NeuroImage}, 219:117044, 2020.

\bibitem{Wolters2004}
C.~H. Wolters, L.~Grasedyck, and W.~Hackbusch.
\newblock {Efficient computation of lead field bases and influence matrix for the {FEM}-based {EEG} and {MEG} inverse problem}.
\newblock {\em Inverse Problems}, 20(4):1099--1116, 2004.

\bibitem{ram06}
C.~Ramon, P.~H. Schimpf, and J.~Haueisen.
\newblock Influence of head models on {EEG} simulations and inverse source localizations.
\newblock {\em Biomed. Eng. Online}, 5(10), 2006.

\bibitem{Sohrabpour2015Electrodenum}
Abbas Sohrabpour, Yunfeng Lu, Pongkiat Kankirawatana, Jeffrey Blount, Hyunmi Kim, and Bin He.
\newblock Effect of {EEG} electrode number on epileptic source localization in pediatric patients.
\newblock {\em Clinical neurophysiology}, 126(3):472--480, 2015.

\bibitem{Montoya2021Electrodenum}
Jair Montoya-Martínez, Jonas Vanthornhout, Alexander Bertrand, and Tom Francart.
\newblock Effect of number and placement of {EEG} electrodes on measurement of neural tracking of speech.
\newblock {\em PloS one}, 16(2):e0246769--, 2021.

\bibitem{Moritz2022Electrodenum}
Moritz Tacke, Katharina Janson, Katharina Vill, Florian Heinen, Lucia Gerstl, Karl Reiter, and Ingo Borggraefe.
\newblock Effects of a reduction of the number of electrodes in the {EEG} montage on the number of identified seizure patterns.
\newblock {\em Scientific reports}, 12(1):4621--7, 2022.

\bibitem{SolerAndres2022ElectrodenumTargeted}
Andres Soler, Luis~Alfredo Moctezuma, Eduardo Giraldo, and Marta Molinas.
\newblock Automated methodology for optimal selection of minimum electrode subsets for accurate {EEG} source estimation based on genetic algorithm optimization.
\newblock {\em Scientific reports}, 12(1):11221--18, 2022.

\bibitem{Horrillo-Maysonnial2023ElectrodenumTargeted}
A.~Horrillo-Maysonnial, T.~Avigdor, C.~Abdallah, D.~Mansilla, J.~Thomas, N.~von Ellenrieder, J.~Royer, B.~Bernhardt, C.~Grova, J.~Gotman, and B.~Frauscher.
\newblock Targeted density electrode placement achieves high concordance with traditional high-density {EEG} for electrical source imaging in epilepsy.
\newblock {\em Clinical neurophysiology}, 156:262--271, 2023.

\bibitem{HeshengLiu2005SsL}
Hesheng Liu, P.H Schimpf, Guoya Dong, Xiaorong Gao, Fusheng Yang, and Shangkai Gao.
\newblock Standardized shrinking {LORETA-FOCUSS (SSLOFO)}: a new algorithm for spatio-temporal {EEG} source reconstruction.
\newblock {\em IEEE transactions on biomedical engineering}, 52(10):1681--1691, 2005.

\bibitem{LahtinenJoonas2026StableSKF}
Joonas Lahtinen.
\newblock Stable {EEG} source estimation for standardized {K}alman filter using rate-of-change tracking.
\newblock {\em Computer methods and programs in biomedicine}, 278:109280--, 2026.

\bibitem{Dilshanie2026SKFParameters}
Dilshanie Prasikala, Joonas Lahtinen, Alexandra Koulouri, and Sampsa Pursiainen.
\newblock The effect of prior parameters on standardized {K}alman filter-based {EEG} source localization.
\newblock {\em Biomedical signal processing and control}, 121:110233--, 2026.

\bibitem{Lahtinen2025SVD}
Joonas Lahtinen and Alexandra Koulouri.
\newblock Enhanced localization and orientation estimations in focal {EEG} source imaging using {SVD}-based coordinate transform.
\newblock {\em Brain Topography}, 38(6):78, Oct 2025.

\end{thebibliography}

\end{document}